\documentclass[lettersize,journal]{IEEEtran}
\usepackage{amsmath,amsfonts}
\usepackage{algorithmic}
\usepackage{algorithm}
\usepackage{array}
\usepackage{textcomp}
\usepackage{stfloats}
\usepackage{url}
\usepackage{verbatim}
\usepackage{graphicx}
\usepackage{cite}
\usepackage{amsfonts}       
\usepackage{color}
\usepackage{amsmath}
\usepackage{amssymb}
\usepackage{mathtools}
\usepackage{amsthm}
\usepackage{subfigure}
\usepackage{multirow}
\usepackage{cases}
\usepackage{hyperref}       
\usepackage{url}            
\makeatletter

\newcommand{\Rmnum}[1]{\expandafter\@slowromancap\romannumeral #1@}
\makeatother

\newtheorem{theorem}{Theorem}

\usepackage{etoolbox}
\DeclareFontFamily{OT1}{pzc}{}
\DeclareFontShape{OT1}{pzc}{m}{it}{<-> s * [1.10] pzcmi7t}{}
\DeclareMathAlphabet{\mathpzc}{OT1}{pzc}{m}{it}

\newcommand{\MyMapTemplatePrefix}[4]{\expandafter#1\csname#3#4\endcsname{#2{#4}}}
\newcommand{\MyMapTemplatePrefixNew}[5]{\expandafter#1\csname#4#5\endcsname{#2{#3{#5}}}}
\forcsvlist{\MyMapTemplatePrefix {\def} {\mathbf} {}} {A,B,C,D,E,F,G,H,I,J,K,L,M,N,O,P,Q,R,S,T,U,V,W,X,Y,Z}
\forcsvlist{\MyMapTemplatePrefix {\def} {\mathbf} {}} {a,b,c,d,e,f,g,h,i,j,k,l,m,n,o,p,q,r,s,t,u,v,w,x,y,z,1,0}
\forcsvlist{\MyMapTemplatePrefix {\def} {\widetilde} {wt}} {A,B,C,D,E,F,G,H,I,J,K,L,M,N,O,P,Q,R,S,T,U,V,W,X,Y,Z}
\forcsvlist{\MyMapTemplatePrefix {\def} {\widetilde} {wt}} {a,b,c,d,e,f,g,h,i,j,k,l,m,n,o,p,q,r,s,t,u,v,w,x,y,z} 
\forcsvlist{\MyMapTemplatePrefixNew {\def} {\widetilde}{\mathbf} {tb}} {A,B,C,D,E,F,G,H,I,J,K,L,M,N,O,P,Q,R,S,T,U,V,W,X,Y,Z}
\forcsvlist{\MyMapTemplatePrefixNew {\def} {\widetilde}{\mathbf} {tb}} {a,b,c,d,e,f,g,h,i,j,k,l,m,n,o,p,q,r,s,t,u,v,w,x,y,z}
\forcsvlist{\MyMapTemplatePrefixNew {\def} {\overleftarrow}{\mathbf}{ol}} {A,B,C,D,E,F,G,H,I,J,K,L,M,N,O,P,Q,R,S,T,U,V,W,X,Y,Z}
\forcsvlist{\MyMapTemplatePrefixNew {\def}{\overrightarrow}{\mathbf}{or}} {A,B,C,D,E,F,G,H,I,J,K,L,M,N,O,P,Q,R,S,T,U,V,W,X,Y,Z}
\forcsvlist{\MyMapTemplatePrefix {\def} {\widehat} {wh}} {A,B,C,D,E,F,G,H,I,J,K,L,M,N,O,P,Q,R,S,T,U,V,W,X,Y,Z}
\forcsvlist{\MyMapTemplatePrefix {\def} {\widehat} {wh}} {a,b,c,d,e,f,g,h,i,j,k,l,m,n,o,p,q,r,s,t,u,v,w,x,y,z}
\forcsvlist{\MyMapTemplatePrefixNew {\def} {\widehat}{\mathbf} {hb}} {A,B,C,D,E,F,G,H,I,J,K,L,M,N,O,P,Q,R,S,T,U,V,W,X,Y,Z}
\forcsvlist{\MyMapTemplatePrefixNew {\def} {\widehat}{\mathbf} {hb}} {a,b,c,d,e,f,g,h,i,j,k,l,m,n,o,p,q,r,s,t,u,v,w,x,y,z}
\forcsvlist{\MyMapTemplatePrefix {\def} {\mathcal}{mc}} {A,B,C,D,E,F,G,H,I,J,K,L,M,N,O,P,Q,R,S,T,U,V,W,X,Y,Z}
\forcsvlist{\MyMapTemplatePrefix {\def} {\mathpzc}{mc}} {a,b,c,d,e,f,g,h,i,j,k,l,m,n,o,p,q,r,s,t,u,v,w,x,y,z}
\forcsvlist{\MyMapTemplatePrefixNew {\def} {\widetilde}{\mathcal} {tc}} {A,B,C,D,E,F,G,H,I,J,K,L,M,N,O,P,Q,R,S,T,U,V,W,X,Y,Z}
\forcsvlist{\MyMapTemplatePrefix {\def} {\mathbb} {mb}} {A,B,C,D,E,F,G,H,I,J,K,L,M,N,O,P,Q,R,S,T,U,V,W,X,Y,Z}
\forcsvlist{\MyMapTemplatePrefix {\DeclareMathOperator} {} {} } {tr,diag,sgn}
\def\glb{\text{glb}}    \def\next{\text{next}}
\newtheorem{Remark}{Remark}

\begin{document}

\title{TVDO: Tchebycheff Value-Decomposition Optimization for Multiagent Reinforcement Learning}

\author{Xiaoliang Hu, Pengcheng Guo, Yadong Li, Guangyu Li, Zhen Cui, Jian Yang
\thanks{
Manuscript received 12 August 2023; revised 23 March 2024 and 15 July 2024; accept 3 September 2024.  
This work was supported by the National Natural Science Foundation of China (Grants Nos. 62476133, 62006119), the fundamental research funds for the central universities under Grant 30919011232, the Natural Science Foundation of Shandong Province (Grant No. ZR2022LZH003). \textit{(Corresponding author: Zhen Cui.)}

The authors are with the School of Computer Science and Engineering, Nanjing University of Science and Technology, Nanjing 210094, China 
(e-mail: \href{peter_hu_xl@njust.edu.cn}{peter\_hu\_xl@njust.edu.cn}, \href{glimmer007@njust.edu.cn}{glimmer007@njust.edu.cn}, \href{liyadong@uzz.edu.cn}{liyadong@uzz.edu.cn}, \href{guangyu.li2017@njust.edu.cn}{guangyu.li2017@njust.edu.cn}, \href{zhen.cui@njust.edu.cn}{zhen.cui@njust.edu.cn}, \href{csjyang@njust.edu.cn}{csjyang@njust.edu.cn})

Digital Object Identifier 10.1109/TNNLS.2024.3455422
}
}

\markboth{IEEE Transactions on Neural Networks and Learning Systems}%
{Shell \MakeLowercase{\textit{et al.}}: A Sample Article Using IEEEtran.cls for IEEE Journals}


\maketitle

\begin{abstract}
In cooperative multiagent reinforcement learning (MARL), centralized training with decentralized execution (CTDE) has recently attracted more attention due to the physical demand. However, the most dilemma therein is the inconsistency between jointly-trained policies and individually-executed actions. In this article, we propose a factorized Tchebycheff value-decomposition optimization (TVDO) method to overcome the trouble of inconsistency. In particular, a nonlinear Tchebycheff aggregation function is formulated to realize the global optimum by tightly constraining the upper bound of individual action-value bias, which is inspired by the Tchebycheff method of multi-objective optimization. 
We theoretically prove that, under no extra limitations, the factorized value decomposition with Tchebycheff aggregation satisfies the sufficiency and necessity of Individual-Global-Max (IGM), which guarantees the consistency between the global and individual optimal action-value function. Empirically, in the climb and penalty game, we verify that TVDO precisely expresses the global-to-individual value decomposition with a guarantee of policy consistency. Meanwhile, we evaluate TVDO in the SMAC benchmark, and extensive experiments demonstrate that TVDO achieves a significant performance superiority over some SOTA MARL baselines.
\end{abstract}


\section{Introduction}
\label{submission}
\IEEEPARstart{W}{ith} the development of deep learning and multiagent system, multiagent reinforcement learning (MARL) has attracted much attention in recent years~\cite{Jie2024Time, Du2023Multiagent, Song2024Local, wu2020multi, hernandez2019survey}, and widely-used to solve a variety of cooperative and competitive tasks, such as video games~\cite{kurach2020google, wurman2022outracing}, robot swarms~\cite{He2022msac}, DOTA2~\cite{2019dotaberner},  transportation~\cite{dong2020event} and autonomous driving~\cite{zhou2020smarts}. In many realistic multiagent settings, an individual agent often only captures local/partial observation or communication due to the limitation of the environment. Hence, the learning of decentralized policies is required on the condition of individual action observation of each agent. The advantage of decentralized learning is two-fold: i) reduces the computation burden of exponential growth in global state and joint action space; ii) facilitates using conventional reinforcement learning methods. In particular, a representative work is independent Q-learning (IQL)~\cite{2017multiagent}, where each agent is trained individually to learn their own policy. However, the decentralized policy may not converge to the global optimum because of the inherent non-stationarity of partial observability~\cite{lowe2017multi}. 

To address the above problem, the popular solution is to take the paradigm of centralized training with decentralized execution (CTDE)~\cite{oliehoek2016concise}, in which decentralized policies could be learned in a simulated centralized setting.
In recent times, a multitude of MARL methods has been introduced for the paradigm of CTDE. In particular, one main stream of works (e.g., VDN~\cite{sunehag2017value}, QMIX~\cite{rashid2018qmix}, QTRAN~\cite{son2019qtran}, Qatten~\cite{yang2020qatten}, QPLEX~\cite{wang2021qplex}, etc.) focused on conceiving new algorithms for value decomposition. The value-based branch that we follow aims to decompose the global action-value function into individual action-value functions to ensure consistency between the global and individual policies.   

We have witnessed much success in value-decomposition MARL~\cite{sunehag2017value,rashid2018qmix,son2019qtran,rashid2020weighted,yang2020qatten,wang2021qplex,sun2021dfac,zohar2022locality, shen2022resq, Wei2024VGN}. As the representative works, VDN~\cite{sunehag2017value} and QMIX~\cite{rashid2018qmix} learn a linear value decomposition by using the additivity and the monotonicity respectively. In particular, VDN may be understood as a special case of QMIX by treating each agent equally in the accumulative action-value function. However, they are implicitly built on the condition of the structure constraint of monotonicity, and only satisfy the sufficient condition of Individual-Global-Max (IGM)~\cite{son2019qtran}. To release the restrictive constraint, QTRAN~\cite{son2019qtran} converts the original global action-value function into a newly factorized function that maintains the optimal global policy. Three complicated networks are designed to fulfill three parts of the factorized action-value function. Although QTRAN well presents a sufficient and necessary condition of IGM, it requires an extra limitation, i.e., affine transformation. More recently, Weighted QMIX~\cite{rashid2020weighted} uses the weighted projection to solve the suboptimal policy problem existed in QMIX, but still requires the monotonicity constraint. QPLEX~\cite{wang2021qplex} transforms equivalently the condition of IGM to the advantage-based IGM condition. However, the method requires the positive importance weights in the duplex dueling component to guarantee the IGM condition. Even with these great progresses in value decomposition, almost the above methods cannot guarantee the complete expressiveness of IGM condition except for QTRAN and QPLEX. Although QTRAN and QPLEX provide theoretical guarantee for factorizing cooperative MARL tasks, they require an extra constraint (e.g. affine transformation (QTRAN) or the positive importance weights (QPLEX)). Until now, solving an equivalent and accessible decomposition to precisely express IGM remains challenging and inevitable to fulfill the consistency between centralized policies and individual actions in MARL.

In this paper, we introduce a novel MARL approach called factorized Tchebycheff Value Decomposition Optimization (TVDO) in the fashion of CTDE, which is inspired by the Tchebycheff approach of multi-objective optimization~\cite{Tchebycheff}. Concretely, we introduce a novel nonlinear aggregation function, based on the Tchebycheff approach, to achieve the global optimum through tightly constraining the upper bound of individual action-value bias. Theoretically, we prove that the proposed TVDO method satisfies the sufficiency and necessity of the IGM condition. Thus, the decomposition way in TVDO represents precisely from global to individual value factorization with a guarantee of policy consistency. Empirically, in the climb and penalty game~\cite{claus1998, panait2006}, we verify that TVDO indeed represents precisely the value factorization, while most existing value-based decomposition MARL methods can not. Although the value-based decomposition MARL approach QTRAN achieves better performance on the learning optimality and stability, it needs the extra condition of an affine transformation from the joint Q-value to individual Q-values, which may result in unsatisfactory performance in some complex scenarios such as StarCraft II. We also evaluate TVDO on a series of decentralized micromanagement scenarios in the SMAC ~\cite{2019starcraft} benchmark. The results show that the performance of the TVDO algorithm is superior to the existing methods in terms of convergence speed. 

In summary, our contribution is threefold: i) proposes the Tchebycheff Value Decomposition method, which is easy to use and effective for MARL; ii) theoretically proves that TVDO satisfies the sufficiency and necessity of the IGM condition with no extra constraint; iii) reports the SOTA performance in the SMAC benchmark.

\section{Related Work}
\label{releated work}
With the development of deep reinforcement learning, a large number of MARL methods have been proposed to solve cooperative multiagent tasks. In this paper, we will discuss value-based, policy-based, and distributed-based reinforcement learning methods in the cooperative multiagent environment.

\subsection{Value-based methods for MARL} 
Several value-based reinforcement learning methods \cite{sunehag2017value, rashid2018qmix, son2019qtran, yang2020qatten, rashid2020weighted, wang2021qplex, sun2021dfac, zohar2022locality, shen2022resq, Wei2024VGN} have been proposed to decomposes the global action-value function into individual action-value functions to ensure consistency between the global and individual policies, i.e., IGM~\cite{son2019qtran}. However, the earlier representative methods~\cite{sunehag2017value, rashid2018qmix, rashid2020weighted, yang2020qatten} can not satisfy the sufficiency and necessity IGM due to the structural constraints, i.e., additivity and monotonicity, or relaxations. Accordingly, QTRAN~\cite{son2019qtran} converts the original global action-value function into a newly factorized function that maintains the optimal global policy. In addition, QPLEX~\cite{wang2021qplex} extended further the action-value function conditions on advantage-based IGM to keep the consistency between global and individual optimal actions. Although QTRAN~\cite{son2019qtran} and QPLEX~\cite{wang2021qplex} provide theoretical guarantees for factorizing cooperative MARL tasks, they require an extra constraint (e.g. affine transformation (QTRAN) or the positive importance weights (QPLEX)). Furthermore, ResQ~\cite{shen2022resq} finds the optimal joint policy for any state-action value function through residual functions and satisfies the IGM condition. More recently, VGN~\cite{Wei2024VGN} proposes a novel value decomposition method to model the relationship between the joint action-value function and the individual action-value functions. Overall, almost all approaches (except for QTRAN, QPLEX, and VGN) still suffer from the simplicity of decomposition and relaxation of these constraints, which may result in poor performance in some complex tasks such as in the SMAC benchmark.

\subsection{Policy-based methods for MARL}
In order to deal with continuous action space, some policy-based methods \cite{lowe2017multi, foerster2018counterfactual, iqbal2019actorattentioncritic, de2020deep, wang2020off, kim2021policy, zhang2021fop} have been proposed in recent years. MADDPG \cite{lowe2017multi} and COMA \cite{foerster2018counterfactual} are the variant of the actor-critic method, which learn a centralized critic instead of an individual critic on the condition of joint action-observation trajectory. Particularly, compared with MADDPG, COMA only learns an actor network by sharing parameters to speed learning. \cite{de2020deep} proposed a novel method called FacMADDPG, which facilitates the critic in decentralized POMDPs based on MADDPG and QMIX. DOP \cite{wang2020off} introduced firstly the value decomposition similar to Qatten \cite{yang2020qatten} into the multiagent actor-critic framework with on-policy TD($\lambda$) and tree backup technique. \cite{kim2021policy} posed a meta-learning multiagent policy gradient theorem to adapt quickly to the non-stationarity of the environment, and gave the theoretical analysis in detail. Furthermore, FOP \cite{zhang2021fop} achieved global optimum by transforming the IGM condition into the Individual-Global-Optimal (IGO) condition. However, due to centralized-decentralized mismatch (CDM) problems, the above methods perform unsatisfactorily compared with value-based methods. 

\subsection{Distributed-based methods for MARL}
Some distributed-based approaches\cite{zhang2018fully, suttle2020multi, sun2021dfac, dai2022distributed, Hu2024Distributed} have been proposed by combining the distributed optimization technique with multiagent learning. Zhang \textit{et al.}~\cite{zhang2018fully} proposed two fully decentralized actor-critic algorithms to maximize the globally averaged return over the network and provide provable convergence guarantees. However, this work is an on-policy algorithm, which implies that each agent learns solely based on the policy it is currently executing. Suttle \textit{et al.}~\cite{suttle2020multi} presented a distributed off-policy actor-critic method for MARL. Despite the recent advances in the field, these methods require the communication graphs to be undirected and the weight matrices to be doubly stochastic. Hence, Dai \textit{et al.}~\cite{dai2022distributed} proposed two distributed actor-critic algorithms for MARL over the directed graph with fixed topology that only require the weight matrix to be row or column stochastic. Furthermore, DFAC \cite{sun2021dfac} introduced a Distributional Value Function Factorization (DFAC) framework by integrating distributional reinforcement learning and exiting value-based decomposition MARL algorithms, e.g. IQL, VDN, and QMIX. However, its DFAC variants (DDN, DMIX) still only satisfy the sufficient condition of IGM.

\section{Problem Description}
The cooperative MARL task can be regarded as a decentralized partially observable Markov decision process (Dec-POMDP)~\cite{oliehoek2016concise}. Formally, the task can be formulated as a tuple ${\langle \mcA, \mcS, \mcO, \mcU, \mcT, P, R, \gamma\rangle}$, where each element is defined as follows:
\begin{itemize}
\item[-] ${\mcA}$: a team of $N$ agents $\{a_1, a_2, \cdots, a_N \}$, in which ${i \in N}$ is the index of the $i$-th agent.
\item[-] ${\mcS}$: a finite set of environmental states. Given a current state ${s \in \mcS}$, we often denote the next state as ${s' \in \mcS}$.
\item[-] ${\mcO}$: a set of observation information of agents. The joint observation at time ${t}$ is denoted as ${o^t} = (o^t_1, o^t_2, \cdots, o^t_N)$, which contains local observation of each agent.
\item[-] ${\mcU}$: a set of joint actions. Let ${u^t} =(u^t_1, u^t_2, \cdots, u^t_N) \in \mcU$ denotes the entire group of agent actions at time ${t}$. Note that the actions may be discrete or continuous.
\item[-] ${\mcT}$: a set of joint action-observation historical trajectories. At time $t$, the action-observation history is ${\tau^{t} = \{o^1, u^1, o^2, u^2, \cdots, o^{t-1}, u^{t-1}, o^t \} \in \mcT}$, which excludes $u^t$ for the next estimation.
\item[-] ${P(s'|s,u)}$: the state transition function, which usually denotes the probability from the state $s$ to the next state $s'$ when taking the action $u$.
\item[-] $R(s,u)$: the reward function. After each agent performs individual action ${u_{i}(i = 1, 2, \cdots, N})$ under the environment state $s$, the team of agents will receive an immediate global reward ${r = R(s, u)}$ shared for all agents.
\item[-] ${\gamma} \in (0, 1]$: the discount factor used in the computation of accumulative return.
\end{itemize}

Further, we introduce two widely used value functions. One is the state value function ${V(s_t)}$ defined as:
${V(s_t) := \mbE_{s,u}[\sum_{k=t}^{\infty}{\gamma^{k-1}R(s_k,u_k)}|u_k \sim \pi(\cdot|s_k)]}$,
where $\pi$ is the policy function. The state value aggregates historical rewards with discount rates. The other is the action-value function ${Q(s_t, u_t)}$ that denotes the value of current state $s_t$ when the joint action $u_t$ is taken by agents.
Conventionally, the action-value function is defined as ${Q(s_t, u_t) := \mathbb{E}_{s_{t+1} \sim P(\cdot|s_t, u_t)}[R(s_t, u_t) + \gamma V(s_{t+1})]}$.
Due to the partial observability that each agent can not obtain fully environmental state information, the state ${s_t}$ is replaced by the action-observation history ${\tau_t}$ in the action-value function ${Q(s_t, u_t)}$. Hereby, the action-value function is approximated to be ${Q(\tau_t, u_t) = \mathbb{E}_{\tau, u} [\wtR_t|\tau_t, u_t]}$, where ${\wtR_t = \sum_{i=0}^\infty {\gamma^{i} r_{t+i}}}$ with $r_{t+i}=R(\tau_{t+i}, u_{t+i})$.

\textbf{Centralized Training with Decentralized Execution}: In realistic multiagent environments, each agent only captures partial observation because of the limitation of the environment. Therefore, the learning of decentralized policies is required on the condition of individual observation of each agent. Perhaps the most naive training method for MARL tasks is IQL~\cite{2017multiagent}, in which each agent is trained independently to learn their policy. However, due to the instability of partial observability, the decentralized policies may not converge to a globally optimal solution under the condition of finite exploration. The most commonly alternative solution is to employ the fashion of centralized training with decentralized execution (CTDE)~\cite{kraemer2016multi}, where each agent learns the policy by optimizing individual action-value function based on the individual observation of each agent and global state in the training phase, and the agent makes its decision with local observation at execution time. 

The bottleneck of CTDE is how to keep the consistency between global and individual policy in the centralized training process when maximizing the joint action-value function. This is the known condition, Individual-Global-Max (IGM)~\cite{son2019qtran}, which is defined as:
\vspace{-0.2cm}
\begin{equation}
    \arg\max\limits_{u}{Q_{\glb}(\tau, u)} = \left(\begin{array}{cc}
        \arg\max\limits_{u_1}{Q_{1}(\tau_1, u_1)} \\
        \vdots \\
        \arg\max\limits_{u_N}{Q_{N}(\tau_N, u_N)}
    \end{array}\right),
    \label{con:IGM}
\end{equation}
where the joint action is defined as $u=(u_1,u_2,\cdots,u_N)$, $Q_\glb(\tau, u)$ denotes the global action-value function, and $Q_i(\tau_i, u_i)$ is the individual action-value function of agent $a_i$.
According to the Eq.~(\ref{con:IGM}), we note that the solution of the greedy global action is tantamount to choosing greedily individual policy under the joint action-observation history $\tau$ for an arbitrary task. In other words, the solved objective is to ensure the consistency between global optimal action and individual optimal actions for CTDE.

\section{Method}
In this section, we propose a novel MARL method called TVDO inspired by the Tchebycheff approach of multi-objective optimization to guarantee the consistency between global and individual optimal actions. In particular, the key idea of our method is to introduce a nonlinear aggregation method for the global policy optimum by tightly constraining the upper bound of individual action-value bias. 

\subsection{Factorized Tchebycheff Value-Decomposition}
To perform CTDE, we take the factorized value decomposition technique line, where the global action-value function could be decomposed with the accumulation of individual action-value functions. As the canonical case in VDN~\cite{sunehag2017value}, $Q_\glb(\tau,u)=\sum_{i=1}^{N}{Q_i(\tau_i,u_i)}$, where $u=(u_1,u_2,\cdots, u_N)$. However, the value factorization technique still suffers structural constraint, namely additive decomposability, which may lead to unsatisfactory performance in some tasks. As discussed in MAVEN~\cite{mahajan2019maven} and QPLEX~\cite{wang2021qplex}, the structure implements sufficient but not necessary conditions for the IGM condition, which limits the representation expressiveness of joint action-value functions. In other words, their full consistency cannot be well guaranteed during learning because of the mismatching between global and local optimal solutions.
It means that there exists a bias between them, denoted as $E$. In the ideal case, $E=0$ indicates the full decomposition consistency, i.e., satisfying the IGM condition. To this end, we minimize the maximum discrepancy of action-value between the global action $u=(u_1,u_2,\cdots, u_N)$ for the team and the local optimal action $\overline{u}_i$ for each agent, formally,
\begin{equation}
    \min_u \left\{E(\tau,u) =  \max\limits_{1\leq i \leq N}{ \left| Q_i(\tau_i,u_i) - Q_i(\tau_i,\overline{u}_i) \right|  }\right\}.
    \label{equ:e_glb}
\end{equation}
Please note that $\overline{u}_i$ is the optimal individual action of agent ${a_i}$, and $Q_i(\tau_i,\overline{u}_i)$ denotes the optimal action-value of agent ${a_i}$ based on individual action-observation history ${\tau_i}$.
When $u_i\rightarrow\overline{u}_i$, there should be $\sum_{i=1}^{N}{Q_i(\tau_i,u_i)} \geq Q_\glb(\tau,u)$ for the accumulated value factorization. Hereby, we propose the final optimization objective as follows:
\vspace{-0.2cm}
\begin{equation}
  \min_{u} \left [ \sum_{i=1}^{N}{Q_i(\tau_i,u_i)} - Q_\glb(\tau,u) + \rho E(\tau,u) \right]^2,
  \label{equ:objective}
\end{equation}
where ${\rho}$ is a weight factor. In the above formula, $\rho E(\tau, u)$ compensates for the inconsistency between the global action-value function $Q_\glb(\tau, u)$ and the summation of individual action-value function $Q_i(\tau_i, u_i)$, which arises from the partial observation of each agent. \textit{The optimization of the above objective can satisfy the sufficiency and necessity of the IGM condition, whose theoretical introduction and proof are deferred to in Section~\ref{sec:theo}.}

\begin{Remark}
The optimized bias term in Eq.~(\ref{equ:e_glb}) is originally inspired by the Tchebycheff approach of multi-objective optimization (MOO)~\cite{Tchebycheff}. But in essence, they are different in the input domain. Here we first review the MOO problem, which optimizes simultaneously two or more objective functions and searches the Pareto-optimal solution~\cite{qian2015constrained}, formally, 
\begin{align}
    \min_x \quad&\{f_1(x), f_2(x), \cdots, f_K(x)\},\\
        \text{s.t.,} \quad &
        g_i(x) \leq 0, \quad i = 1, 2, \cdots, M, \\
        & h_j(x) = 0, \quad j = 1, 2, \cdots, L,
    \label{con:MOOP}
\end{align}
where function ${f_k(x)}$ is the k-th objective function and ${k = 1, 2, \cdots, K}$. The constraints of the MOO problem consist of ${M}$ inequality constraints ${[g_i(x) \leq 0]_{i=1}^{M}}$ and ${L}$ equality constraints ${[h_j(x) = 0]_{j=1}^{L}}$. It is worth noting that the solution $x$ is shared by all objective functions. 

To solve the MOO problem, the Tchebycheff approach~\cite{Tchebycheff} is proposed to transform the original MOO problem into a single-objective Optimization (SOO) counterpart by a nonlinear aggregation function, i.e., the optimal objective is to suppress the bias:
$\max\limits_{1\leq k \leq K}{\left| f_k(x) - f_k^* \right|}$, where $f_k^*$ is the best accessible value of the $k$-th function.
Thus, the Tchebycheff MOO has the idea that seeks a tight upper bound of all objective bias and optimizes the upper bound. In this work, we adapt the idea into Eq.~(\ref{equ:e_glb}).  Obviously, the bias function in Eq.~(\ref{equ:e_glb})) uses different domains, while the MOO function shares a domain. \textbf{In other words, our method does not fall into the category of MOO, so the Pareto-optimal solution is not yet suitable here. Importantly, we find that the multiagent optimization objective in Eq.~(\ref{equ:objective}) with the revised Tchebycheff error term in Eq.~(\ref{equ:e_glb}) works well as verified in experiments and guaranteed in theory.}
\end{Remark}

\subsection{Theoretical Guarantee of Value Decomposition}
\label{sec:theo}
For a given joint action-observation historical trajectory $\tau$ and action $u$, we can note that the global action-value function $Q_\glb(\tau,u)$ for any arbitrarily factorizable MARL tasks can be decomposed into individual action-value functions ${[Q_i(\tau_i,u_i)]_{i=1}^N}$ by the definition of IGM condition. In Theorem~\ref{the:TVDO}, we provide a theoretical guarantee of such a value decomposition, which satisfies the sufficiency and necessity of the IGM condition. Let $\overline{u}$ denote the collection of optimal individual action $[\overline{u}_i]^N_{i=1}$ and $\overline{u}_i = \arg\max\limits_{u_i}{Q_i(\tau_i,u_i)}$. We illustrate the theorem and its detailed proof below.

\begin{theorem}
    The global action-value function $Q_\glb(\tau,u)$ for any arbitrarily factorizable MARL tasks is factorized by ${[Q_i(\tau_i,u_i)]_{i=1}^N}$, if
    \begin{subequations}
       \begin{numcases}
           {\sum_{i=1}^{N}{Q_i(\tau_i,u_i)} - Q_\glb(\tau,u) + \rho E(\tau,u) = }
            0, u = \overline{u}, \label{con:TVDO1} \\
            \geq 0, u \neq \overline{u}, \label{con:TVDO2} 
       \end{numcases} \label{con:TVDO}
    \end{subequations}
    where
    \begin{equation}
        \vspace{-0.3cm}
        \left\{
            \begin{array}{c c}
                 \vspace{0.3cm}
                 E(\tau,u) = \max\limits_{1\leq i \leq N}{ \{ \left| Q_i(\tau_i,u_i) - Q_i(\tau_i,\overline{u}_i) \right| \} }. & \\
                 \vspace{0.2cm}
                 \frac{Q_{\glb}(\tau, \overline{u}) - \sum_{i=1}^{N}{Q_i(\tau_i, u_i)}}{\max\limits_{1\leq i \leq N}{\{ \left| Q_i(\tau_i,u_i) - Q_i(\tau_i,\overline{u}_i) \right| \}}} \leq \rho. & \\
                 \frac{\sum_{i=1}^N{\{ \left| Q_i(\tau_i,u_i) - Q_i(\tau_i,\overline{u}_i) \right| \}}}{\max\limits_{1\leq i \leq N}{\{ \left| Q_i(\tau_i,u_i) - Q_i(\tau_i,\overline{u}_i) \right| \}}} \geq \rho.
            \end{array}
        \right. \label{con: e_g} 
    \end{equation}
    \label{the:TVDO}
\end{theorem}

\begin{proof}
    Sufficiency: \textbf{Theorem \ref{the:TVDO}} indicates that the condition (\ref{con:TVDO}) can derive {IGM}. Therefore, for given individual action-value function $Q_i(\tau_i, u_i)$ that satisfies Eq.~(\ref{con:TVDO}), we will show that $\arg\max\limits_{u}{Q_{\glb}(\tau, u)} = \overline{u}$. That is, we need to prove $Q_\glb(\tau, \overline{u}) \geq Q_{\glb}(\tau, u)$.
    \begin{equation}
        \begin{split}
            Q_\glb(\tau,\overline{u}) &= \sum_{i=1}^{N}{Q_i(\tau_i,\overline{u}_i)} + \rho E(\tau, \overline{u}) \quad (From~(\ref{con:TVDO1}))  \\
            &= \sum_{i=1}^{N}{Q_i(\tau_i,\overline{u}_i)} + \rho \max\limits_{1\leq i \leq N}{\{\left| Q_i(\tau_i,\overline{u}_i) - Q_i(\tau_i,\overline{u}_i) \right|\}} \\
            &= \sum_{i=1}^{N}{Q_i(\tau_i,\overline{u}_i)} + \sum_{i=1}^{N}{Q_i(\tau_i,u_i)} - \sum_{i=1}^{N}{Q_i(\tau_i,u_i)} \\
            &= \sum_{i=1}^{N}{Q_i(\tau_i,u_i)} + \sum_{i=1}^{N}{\{ \left| Q_i(\tau_i,\overline{u}_i) - Q_i(\tau_i,u_i) \right |\}} .
        \end{split} \label{equ: suf_TVDO_1}
    \end{equation}
    According to the condition (\ref{con: e_g}) in Theorem \ref{the:TVDO}, the value of $\rho$ is less than or equal to $\frac{\sum_{i=1}^N{\{ \left| Q_i(\tau_i,u_i) - Q_i(\tau_i,\overline{u}_i) \right| \}}}{\max\limits_{1\leq i \leq N}{\{ \left| Q_i(\tau_i,u_i) - Q_i(\tau_i,\overline{u}_i) \right| \}}}$, we can find that $\sum_{i=1}^N{\{ \left| Q_i(\tau_i,u_i) - Q_i(\tau_i,\overline{u}_i) \right| \}}$ is greater than or equal to $\rho \cdot \max\limits_{1\leq i \leq N}{\{ \left| Q_i(\tau_i,u_i) - Q_i(\tau_i,\overline{u}_i) \right| \}}$.
    \begin{equation}
        \begin{split}
            Q_\glb(\tau,\overline{u}) &= \sum_{i=1}^{N}{Q_i(\tau_i,u_i)} + \sum_{i=1}^{N}{\{ \left| Q_i(\tau_i,\overline{u}_i) - Q_i(\tau_i,u_i)\right| \}} \\
            &\geq \sum_{i=1}^{N}{Q_i(\tau_i,u_i)} + \rho \max\limits_{1\leq i \leq N}{\{ \left| Q_i(\tau_i,u_i) - Q_i(\tau_i,\overline{u}_i) \right| \}}
        \end{split} \label{equ: suf_TVDO_2}
    \end{equation}
    \begin{equation}
        \begin{split}
            Q_\glb(\tau,\overline{u}) &\geq \sum_{i=1}^{N}{Q_i(\tau_i,u_i)} + \rho E(\tau, u) \\
            &\geq Q_{\glb}(\tau, u). \quad (From~(\ref{con:TVDO2}))
        \end{split} \label{equ: suf_TVDO_3}
    \end{equation}
    It means that the collection of individual optimal action of each agent can maximize the global action-value function $Q_\glb(\tau, u)$, illustrating that individual action-value function $Q_i(\tau_i, u_i)$ satisfies the IGM condition.

    Necessity: \textbf{Theorem \ref{the:TVDO}} shows that if the IGM condition holds, then individual and global action-value functions satisfy (\ref{con:TVDO}). By definition of {IGM} condition, we know that if the joint action-value function $Q_\glb(\tau, u)$ is factorized by the individual action-value function $Q_i(\tau_i, u_i)$ for any cooperative MARL tasks, then the followings hold: (i) $Q_{\glb}(\tau, \overline{u}) = \max\limits_{u}{Q_\glb(\tau, u)}$, (ii) $Q_\glb(\tau, u) \leq Q_\glb(\tau,\overline{u})$. Furthermore, $\max\limits_{1\leq i \leq N}{\{ \left| Q_i(\tau_i,u_i) - Q_i(\tau_i,\overline{u}_i) \right| \}} \geq 0$ holds. 
    
    Let $\Gamma = \sum_{i=1}^{N}{Q_i(\tau_i,u_i)} - Q_\glb(\tau,u) + \rho E_(\tau,u)$, we will prove that $\Gamma \geq 0$ for any joint action of agents.
    \begin{equation}
        \begin{split}
           \Gamma &= \sum_{i=1}^{N}{Q_i(\tau_i,u_i)} - Q_\glb(\tau,u) + \rho E(\tau,u) \\
            &\geq \sum_{i=1}^{N}{Q_i(\tau_i,u_i)} - \max_{u}{Q_\glb(\tau,u)} + \rho E(\tau,u)  \\
            &= \sum_{i=1}^{N}{Q_i(\tau_i,u_i)} - Q_\glb(\tau, \overline{u}) + \rho E(\tau,u).
        \end{split}
    \end{equation}

    Since $\rho \geq \frac{Q_{\glb}(\tau, \overline{u}) - \sum_{i=1}^{N}{Q_i(\tau_i, u_i)}}{\max\limits_{1\leq i \leq N}{\{ \left| Q_i(\tau_i,u_i) - Q_i(\tau_i,\overline{u}_i) \right| \}}}$ in (\ref{con: e_g}), we can find that $\rho \max\limits_{1\leq i \leq N}{\{ \left| Q_i(\tau_i,u_i) - Q_i(\tau_i,\overline{u}_i) \right| \}}$ is greater than or equal to $Q_{\glb}(\tau, \overline{u}) - \sum_{i=1}^{N}{Q_i(\tau_i, u_i)}$. Thus, it follows that

    \begin{equation}
        \begin{split}
           \Gamma &\geq \sum_{i=1}^{N}{Q_i(\tau_i,u_i)} - Q_\glb(\tau, \overline{u}) + \rho E(\tau,u) \\
            &\geq \sum_{i=1}^{N}{Q_i(\tau_i,u_i)} - Q_\glb(\tau, \overline{u}) + [Q_{\glb}(\tau, \overline{u}) - \sum_{i=1}^{N}{Q_i(\tau_i, u_i)}] \\
            &= 0.
        \end{split}
    \end{equation}

    Since the factorizable global action-value function and the IGM condition, it suffices to prove that (\ref{con:TVDO}) holds for any joint action $u$. This completes the proof.
\end{proof}

According to Theorem~\ref{the:TVDO}, we can observe that the optimal solution in the objective function in Eq.~(\ref{equ:objective}) is zero, i.e., the case $u=\overline{u}$ in the condition~(\ref{con:TVDO}). At the same time, we limit the weight factor $\rho$ in a bounded range, whose derivation is deferred in the Appendix~\ref{app:derivation}. Furthermore, the lower bound is always lower than or equal to the upper bound for the weight factor, whose proof is deferred in the Appendix~\ref{app:proof}. This theorem indicates that, if the condition in Eq.~(\ref{con:TVDO}) holds, then the global action-value function could be decomposed into the summation of individual action-value functions, i.e., the IGM condition. 
The reason is that when global action is the collection of optimal individual actions $u = (\overline{u}_1, \overline{u}_2, \cdots, \overline{u}_N)$, the discrepancy of action-value between the global action and individual optimal policies will be eliminated, i.e, $E(\tau, u) = 0$, thus  $Q_\glb(\tau,u)=\sum_{i=1}^{N}{Q_i(\tau_i,u_i)}$. It implies that centralized policies could be fully consistent with decentralized policies. In the case of a non-optimal solution, 
the bias $E\neq 0$ aims to compensate for the error of action-value decomposition. 
Under the condition in the Theorem~\ref{the:TVDO}, we could optimize the factorized Tchebycheff value-decomposition by parameterizing the action-value function $Q$ with a neural network, which will be introduced in the section~\ref{sec:tvdo}.

\subsection{Tchebycheff Value Decomposition Optimization}
\label{sec:tvdo}
To handle the above objective function, we design a deep reinforcement learning network framework to enable an end-to-end learning of policy. According to the above theorem, we need to compute the value of $Q$ and estimate the bound for setting a proper $\rho$. For the computation of $Q$, we take the architecture of VDN~\cite{lowe2017multi} as the backbone. For the setting of $\rho$, we simply use its upper bound which works well in practice, because the lower bound depending on the terminal actions is intractable to estimate. The detail is illustrated below.

For each agent ${i}$, we learn an independent agent network to estimate the individual action-value function $Q_i(\tau_i, u_i,\theta_i)$, parameterized by $\theta_i$. It takes the individual action-observation historical trajectory as input at time $t$, i.e., $\tau_i=(o^1_i, u^t_i, \cdots, o^{t-1}_i, u_i^{t-1}, o^t_i)$, and then estimate an action-value vector $Q_i$ w.r.t all actions. Accordingly, at the stage of policy selection, we could choose an optimal $u_i^t$ with the ${\epsilon}$-greedy way. To estimate the global action-value $Q_\glb$, we design a joint action-value network parameterized by $\phi$.
To learn an efficient decomposition of global action-value, we stabilize the learning by calculating the bias term $E$ for the constraint in Eq.~(\ref{con:TVDO}). Thus the accumulation of individual action-value could approximate the joint action-value, and the optimal actions derived from them are identical. In addition, we use the double Q-value network idea introduced in DQN~\cite{2015human} that the parameters of the target network are frozen for a fixed number of steps while updating the main network.

Besides, there is a challenge in choosing the parameter~$\rho$. According to {Theorem \ref{the:TVDO}}, the parameter $\rho$ should be limited in a certain range for conforming the IGM condition.
Since the optimal global action-value function $Q_\glb(\tau, \overline{u})$ can not be precisely estimated during training, i.e., the lower bound cannot be computed, thus we directly use the upper bound of the range to define the parameter $\rho$. Theorem~\ref{the:TVDO} indeed holds only if the value of weight factor $\rho$ is in the range of bounds. In other words, sufficiency and necessity are both satisfied even in the case of the upper bound. Further, we employ the momentum update way for $\rho$, formally,
\begin{equation}
    {\rho} \leftarrow \alpha\rho + (1-\alpha) \frac{\sum_{i=1}^N{\{ \left| Q_i(\tau_i,u_i) - Q_i(\tau_i,\overline{u}_i) \right| \}}}{\max\limits_{1\leq i \leq N}{\{ \left| Q_i(\tau_i,u_i) - Q_i(\tau_i,\overline{u}_i) \right| \}}} ,
    \label{equ:rho}
\end{equation}
where $\alpha$ means a momentum coefficient.

Based on the above idea, TVDO minimizes the following loss to train the individual action-value network $\theta_i$ in the paradigm of \textbf{CTDE}:
\begin{equation}
    \mathcal{L}_{td}(;\theta) = [\sum_{i=1}^{n}Q_i(\tau_i,u_i,\theta_i) - Q_\glb(\tau,u,\phi) + \rho E(\tau, u)]^2 ,
    \label{equ:loss}
\end{equation}
where
\begin{equation}
    E(\tau,u) = \max\limits_{1\leq i \leq N}{\left| Q_i(\tau_i,u_i,\theta_i) - Q_i(\tau_i,\overline{u}_i,\overline{\theta_i}) \right|}.
    \label{equ:obejective fun}
    \nonumber
\end{equation}
where $\theta$ and $\overline{\theta}$ represent the parameters of main and target networks, respectively. The global action-value network $\phi$ is optimized by minimizing the squared TD error. Under the factorizable action-value function in Theorem \ref{the:TVDO}, the group of individual policies converges to the global optimum. For the sake of comprehensiveness, we provide the training process for the proposed TVDO method in Algorithm~\ref{alg:TVDO}. \textit{More detail of the network and hyperparameters will be given in Section~\ref{sec:network}.}

\begin{table*}[t]
    \vspace{-0.3cm}
    \centering
    \renewcommand{\arraystretch}{1.1}
    \caption{The comparison of all value-decomposition baselines and our method.}
    \begin{tabular}{c c c}
        \hline
        \textbf{Method} & \textbf{Decomposition structure} & \textbf{condition for IGM}\\ 
        \hline
        VDN~\cite{lowe2017multi} & additivity & sufficient\\
        QMIX~\cite{rashid2018qmix} & monotonicity & sufficient\\
        QTRAN~\cite{son2019qtran} & soft regularization & sufficient and necessary with extra constraint\\
        Weighted QMIX~\cite{rashid2020weighted} & weighted projection & sufficient\\
        Qatten~\cite{yang2020qatten} & linear constraint & sufficient\\
        QPLEX~\cite{wang2021qplex} & duplex dueling & sufficient and necessary with extra constraint\\
        \textbf{TVDO (ours)} & \textbf{Tchebycheff bias} & \textbf{sufficient and necessary} \\
        \hline
    \end{tabular}
    \label{t: constraint for methods}
    \vspace{-0.3cm}
\end{table*}

\begin{algorithm}[!htb]
    \caption{TVDO}
    \label{alg:TVDO}
    \begin{algorithmic}[1]
        \REQUIRE the set of individual observation $\{o_i\}^N_{i=1}$ and action $\{u_i\}^N_{i=1}$, discount factor $\gamma$, decay $\lambda$, $\rho$, $\epsilon$   \\
        \ENSURE individual action-value networks $\{Q_{\theta_i}\}^n_{i=1}$ and global action-value network $Q_\glb$ with parameter $\phi$;\\
        \STATE Initializing: replay buffer D, the individual current and target individual action-value networks with random parameters $\{\theta_i\}^N_{i=1}$, $\{\overline{\theta}_i\}^N_{i=1}$;
        \FOR{episode = 1 to max-training-episode}
            \STATE Initialize the environment;
            \FOR{${t = 1}$ to max-episode-length}
                \FOR{each agent ${\{a_i\}^N_{i=1}}$}
                    \STATE Get individual action-value $\!Q_i\!$ by feeding $\tau^t_i \!\!=\!\!\! \{o^1_i, u^1_i, \cdots, o^{t-1}_i, u^{t-1}_i, o^t_i \}$ into current action-observation  network $Q_{\theta_i}$;
                    \STATE Select a random action $u^t_i$ within the probability $\epsilon$, otherwise select action $u^t_i = \arg\max \limits_{u_{i}}{Q_{\theta_{i}}(\tau_i, u_i)}$;
                 \ENDFOR
                 \STATE Execute actions $u^t = (u^t_1, u^t_2, \cdots, u^t_N)$ to obtain the environment reward $r^t$ and the observation $o^{t+1}_i$;
                 \STATE Store $(o^t, u^t, r^t, o^{t+1}, done^t)$ in replay buffer D;
            \ENDFOR
            \FOR{agent ${i = 1}$ to ${N}$}
                \STATE Sample a random minibatch of $M$ samples from D:$(o_m, u_m, r_m, o^{\next}_m, done_m)$;
                \STATE Update $\theta_i$ by minimizing the loss in Eq.~(\ref{equ:loss});
            \ENDFOR
            \STATE Update $\phi$ by minimizing the square TD error;
            \IF{step\%$C$ == 0}
            \STATE Update target action-value network: $\overline{\theta}_i = \theta_i$;
            \STATE Update the parameter $\rho$ according to Eq.~(\ref{equ:rho});   
            \ENDIF
        \ENDFOR
    \end{algorithmic}
\end{algorithm}

\section{Discussion: Ours vs Previous Decomposition}
As presented in Section~\ref{releated work}, a lot of value decomposition approaches have been proposed recently for any factorizable MARL tasks, such as VDN~\cite{lowe2017multi}, QMIX~\cite{rashid2018qmix}, QTRAN~\cite{son2019qtran}, Weighted QMIX~\cite{rashid2020weighted}, Qatten~\cite{yang2020qatten}, and QPLEX~\cite{wang2021qplex}. As representative works, VDN~\cite{sunehag2017value} and QMIX~\cite{rashid2018qmix} learn a linear value decomposition by using the additivity and the monotonicity. However, both of them are built on the condition of the structure constraints, and only satisfy the sufficient condition of IGM. In contrast, QTRAN~\cite{son2019qtran} presents a novel factorization approach to release the restrictive structural constraint. However, QTRAN requires an extra limitation in the affine transformation, while it provides a sufficient and necessary condition of IGM. In addition, Weighted QMIX~\cite{rashid2020weighted} uses weighted projection that places more importance on better joint actions to solve the suboptimal policy problem existing in QMIX, but still needs the monotonicity constraint. As discussed in QPLEX~\cite{wang2021qplex}, Qatten~\cite{yang2020qatten} is an extensive work of VDN, which approximately estimates the joint action-value function by incorporating the multi-head attention mechanism to the learning of Q-value mixed network. More recently, QPLEX~\cite{wang2021qplex} introduces the duplex dueling structure to synchronize the action selection between the global and individual action-value functions. In essence, these methods use different forms to establish the relationship between the global action-value function and the individual counterpart. The comparison between the above methods and our TVDO method is shown in Table~\ref{t: constraint for methods}. We can notice that all approaches only present the sufficient condition of IGM except for QTRAN~\cite{son2019qtran}, QPLEX~\cite{wang2021qplex}, and our method. Although QTRAN achieves better performance on the learning optimality and stability, it needs the extra condition of an affine transformation, which may result in unsatisfactory performance in some complicated and real scenarios such as StarCraft \Rmnum{2}. Furthermore, QPLEX also requires an extra constraint, i.e., the positive importance weights, while it provides the theoretical guarantee of the IGM condition. 

Particularly, \textbf{the similarities} between ours and QTRAN~\cite{son2019qtran} are: \textbf{(i)} The value decomposition takes the linear weighting style, also often used in previous works such as VDN and QMIX, thus ours and QTRAN look similar in the form of value factorization, but which is not our contribution. \textbf{(ii)} For a fair comparison, we use some common basic network units (e.g., individual action-value network), which are framed in the unit code scheme to conveniently evaluate those classic methods. Although our method looks similar to QTARN, they are different. More importantly, \textbf{the differences} between ours and QTRAN are: (i) The motivation: our method introduces the bias inspired by the Tchebycheff approach of multi-objective optimization, while QTRAN directly derives a discrepancy term between global and individual policies. (ii) The constraint terms, i.e., the optimized bias $E_\glb(\tau, u)$ in ours and the discrepancy $V_{\glb}(\tau, u)$ in QTRAN, are different in design rules, and the theoretical proofs are also distinct. (iii) The bias $E_\glb(\tau, u)$ could be calculated exactly, while the discrepancy $V_{\glb}(\tau, u)$ needs to be estimated by one network. (iv) Ours satisfies the necessity and sufficiency of the IGM condition without any extra limitations, while QTRAN needs the extra constraint of an affine transformation.

\begin{figure*}[t]
    \centering
    \subfigure[The payoff of climb and penalty game]{
        \includegraphics[width=5.8cm]{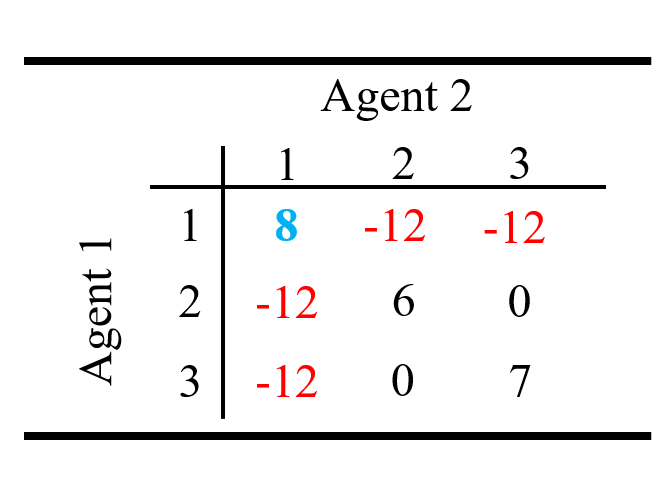}
        \label{fig:payoff_matrices_climb_and_penalty_game}
    }
    \subfigure[Median Episode Reward]{
        \includegraphics[width=5.4cm]{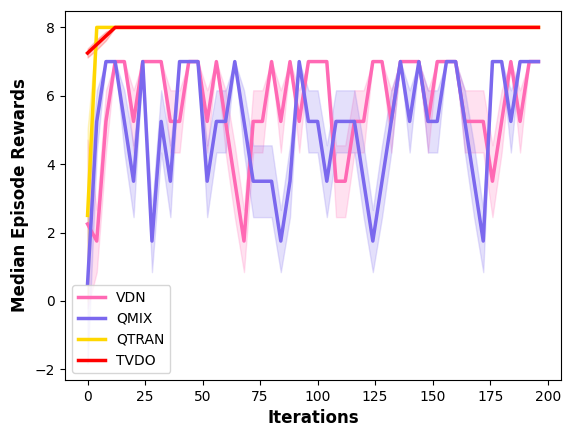}
        \label{fig:episode_rewards_TVDO_vs_VDN_QMIX}
    }
    \subfigure[Median Episode Reward]{
        \includegraphics[width=5.2cm]{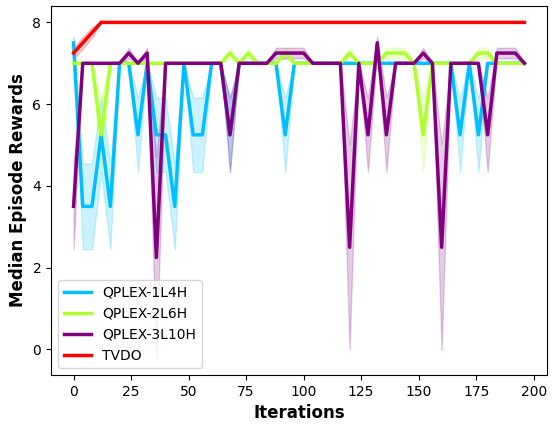}
        \label{fig:episode_rewards_TVDO_vs_QPLEX}
    }
    \caption{(a) The payoff matrices of 3 $\times$ 3 climb and penalty game for 2 agents. Both agents gain the same payoff for a joint action. (b) The median episode rewards of our method(TVDO) vs VDN, QMIX, and QTRAN. (c) The median episode rewards of our method(TVDO) vs QPLEX. In particular, QPLEX with $a$L$b$H denotes the network with $a$ layers and $b$ heads (multi-head attention), respectively.}
    \label{fig:climb_and_penalty_game_result}
\end{figure*}

\section{Climb and Penalty Game}
To illustrate the complete representation capacity of our method compared with existing value decomposed MARL algorithms including VDN~\cite{sunehag2017value}, QMIX~\cite{rashid2018qmix}, QTRAN~\cite{son2019qtran}, QPLEX~\cite{wang2021qplex}, we consider the climb and penalty game (or Matrix Game) from previous literature~\cite{claus1998, panait2006}. This is a simple single-stage cooperative game for 2 agents, which is shown in Figure ~\ref{fig:payoff_matrices_climb_and_penalty_game}. In specific, agent 1 and 2 each have three actions at their disposal. Agents in the climb domain receive maximum payoff (as the blue number in Figure ~\ref{fig:payoff_matrices_climb_and_penalty_game}) when both agents select action 1 (the form of joint action is referred to (1, 1)). However, the team reward matrix has a second equilibrium when they both choose action 3. Due to that the joint reward of (3, 3) is lower than at (1, 1), and the joint action is a suboptimal equilibrium. Moreover, the joint action (2, 2) is a third equilibrium. Additionally, agents obtain a vital penalty when they choose the joint action (1, 2), (1, 3), (2, 1) or (3, 1). Therefore, agents aim to select the optimal joint actions in the game. In particular, the difficulty of this game arises from the penalties incurred when joint actions are not coordinated effectively. Furthermore, the presence of suboptimal collaborations that manage to avoid penalties adds an additional challenge to the game.

We train our method TVDO and other value-decomposed MARL algorithms on this game for 80,000 episodes and examine the final value functions in the limit of complete exploration ($\epsilon$ = 1). Specifically, complete exploration is to ensure that each approach explores all game states. Furthermore, individual action-value function networks consist of two hidden layers, which are shared across all baselines. The global action-value network $Q_\glb$ is composed of two hidden layers, each consisting of 32 units and ReLU non-linearities. All neural networks are trained using the Adam optimizer. The full median episode rewards results of the proposed TVDO method and some MARL baselines are shown in Figure ~\ref{fig:episode_rewards_TVDO_vs_VDN_QMIX} and~\ref{fig:episode_rewards_TVDO_vs_QPLEX}, which show that only TVDO and QTRAN can achieve the optimal performance, while other MARL methods (VDN, QMIX, QPLEX) fall into the suboptimum because of penalty associated with miscoordinated actions. Although QTRAN achieves better performance on the learning optimality and stability, it needs the extra condition of an affine transformation from the joint action-value to individual action-values, which may result in unsatisfactory performance in some complex environments such as StarCraft \Rmnum{2}. In particular, QPLEX introduces a scalable multi-head attention module with different heads of attention and layers (e.g. QPLEX-1L4H, QPLEX-2L6H, QPLEX-3L10H) to learn importance weight. Figure~\ref{fig:episode_rewards_TVDO_vs_QPLEX}, which shows the learning curves of TVDO and QPLEX, demonstrates that TVDO can converge to the optimum whereas QPLEX suffers from the learning optimality and stability while it performs better by increasing the scale of neural network.

\section{Experiments}
\label{sec:exp}
In this section, we use the StarCraft Multi-Agent Challenge (SMAC)~\cite{2019starcraft} benchmark to experimentally evaluate the performance of TVDO. 
All experiments adopt the default settings and are conducted on a 2.90GHz Intel Core i7-10700 CPU, 64G RAM, and NVIDIA GeForce RTX 3090 GPU. Note that all results are based on four training runs with different random seeds in the experiments. In addition, the version of StarCraft \Rmnum{2} used in this work is SC2.4.6.2.69232, which is the same version used as SMAC~\cite{2019starcraft}. 

\subsection{Experimental Setup}
In StarCraft \Rmnum{2}, agents, which select actions that condition on local observation in the limited field of view by a MARL approach, compete against an integrated game AI, striving to defeat their opponents. We perform experiments on a collection of StarCraft \Rmnum{2} micromanagement scenarios, categorized into three levels of difficulty: \textit{Easy, Hard,} and \textit{Super-Hard}. The \textit{Easy} category comprises the following scenarios: 1c3s5z, 2s\_vs\_1sc, 2s3z, 3s\_vs\_3z, 3s5z and 8m. The \textit{Hard} category includes 2c\_vs\_64zg, 3s\_vs\_5z, 5m\_vs\_6m, 10m\_vs\_11m, 25m and bane\_vs\_bane. The \textit{Super-Hard} category encompasses 6h\_vs\_8z, 27m\_vs\_30m, MMM2 and so\_many\_baneling.
The exhaustive list of challenges is presented in Table~\ref{table1}.

\begin{table*}[htbp]
    \centering
    \renewcommand{\arraystretch}{1.1}
    \caption{The StarCraft \Rmnum{2} multiagent challenge [SMAC~\cite{2019starcraft}].}
    \begin{tabular}{c c c c}
        \hline
        \textbf{Map Name} & \textbf{Category} & \textbf{Ally Units} & \textbf{Enemy Units} \\ 
        \hline
        1c3s5z & \multirow{6}{*}{\textit{Easy}} & 1 Colossus, 3 Stalkers \& 5 Zealots & 1 Colossus, 3 Stalkers \& 5 Zealots \\
        2s\_vs\_1sc & & 2 Stalkers & 1 Spine Crawler \\
        2s3z & & 2 Stalkers \& 3 Zealots & 2 Stalkers \& 3 Zealots \\
        3s\_vs\_3z & & 3 Stalkers & 3 Zealots \\
        3s5z & & 3 Stalkers \& 5 Zealots & 3 Stalkers \& 5 Zealots \\
        8m & & 8 Marines & 8 Marines \\
        \hline
        2c\_vs\_64zg & \multirow{6}{*}{\textit{Hard}} & 2 Colossi & 64 Zerglings \\
        3s\_vs\_5z & & 3 Stalkers & 5 Zealots \\
        5m\_vs\_6m & & 5 Marines & 6 Marines \\
        10m\_vs\_11m & & 10 Marines & 11 Marines \\
        25m & & 25 Marines & 25 Marines\\
        bane\_vs\_bane & & 20 Zerglings \& 4 Banelings & 20 Zerglings \& 4 Banelings \\
        \hline
        6h\_vs\_8z & \multirow{5}{*}{\textit{Super-Hard}} & 6 Hydralisks & 8 Zealots \\
        27m\_vs\_30m & & 27 Marines & 30 Marines \\
        MMM2 & & 1 Medivac, 2 Marauders \& 7 Marines & 1 Medivac, 3 Marauders \& 8 Marines \\
        so\_many\_baneling & & 7 Zealots & 32 Banelings \\
        \hline
    \end{tabular}
    \label{table1}
\end{table*}

\begin{table}[htbp]
    \centering
    \renewcommand{\arraystretch}{1.1}
    \caption{The component size of observation and action for all scenarios.}
    \begin{tabular}{c c c c c c}
        \hline
        \textbf{Map Name} & \textbf{Move} & \textbf{Enemy} & \textbf{Ally} & \textbf{Own} & \textbf{Attack\_id}\\ 
        \hline
        1c3s5z & 4 & (9, 9) & (8, 9) & 5 & 9 \\
        2s\_vs\_1sc & 4 & (1, 5) & (1, 6) & 2 & 1  \\
        2s3z & 4 & (5, 8) & (4, 8) & 4 & 5\\
        3s\_vs\_3z & 4 & (3, 6) & (2, 6) & 2 & 3 \\
        3s5z & 4 & (8, 8) & (7, 8) & 4 & 8\\
        8m & 4 & (8, 5) & (7, 5) & 1 & 8\\
        \hline
        2c\_vs\_64zg & 4 & (64, 5) & (1, 6) & 2 & 64 \\
        3s\_vs\_5z & 4 & (5, 6) & (2, 6) & 2 & 5\\
        5m\_vs\_6m & 4 & (6, 5) & (4, 5) & 1 & 6\\
        10m\_vs\_11m & 4 & (11, 5) & (9, 5) & 1 & 11\\
        25m & 4 & (25, 5) & (24, 5) & 1 & 25 \\
        bane\_vs\_bane & 4 & (24, 7) & (23, 7) & 3 & 24  \\
        \hline
        6h\_vs\_8z & 4 & (8, 6) & (5, 5) & 1 & 8 \\
        27m\_vs\_30m & 4 & (30, 5) & (26, 5) & 1 & 30   \\
        MMM2 & 4 & (12, 8) & (9, 8) & 4 & 12 \\
        so\_many\_baneling & 4 & (32, 5) & (7, 6) & 2 & 32  \\
        \hline
    \end{tabular}
    \label{table2}
\end{table}
Agents get local observations $\{o_i\}^N_{i=1}$, which are composed of agent movement, enemy, ally, and agent unit features, in the range of their sight at each time. The dimensionality of the observation vector may fluctuate, contingent upon the specific environment configuration and the types of units existing within the scenario. For example, non-Protoss units typically lack shields, and the inclusion of movement features such as terrain height and pathing grid may vary. Additionally, the unit\_type is excluded if there is only one type of unit present in the map. Agent movement features denote the ability to move in the cardinal directions of north, east, south, and west. In particular, the feature vector of $o_i$ encompasses the following attributes for both allied and enemy units: health, unit\_type, shield, relative x, relative y, and distance. The features of the agent unit contain its shield, health, and unit\_type. All features are normalized by using Min-Max normalization. 

\begin{figure*}[t]
    \centering
    \setcounter{subfigure}{0}
    \subfigure[Average median win rates]{
        \includegraphics[width=6.5cm]{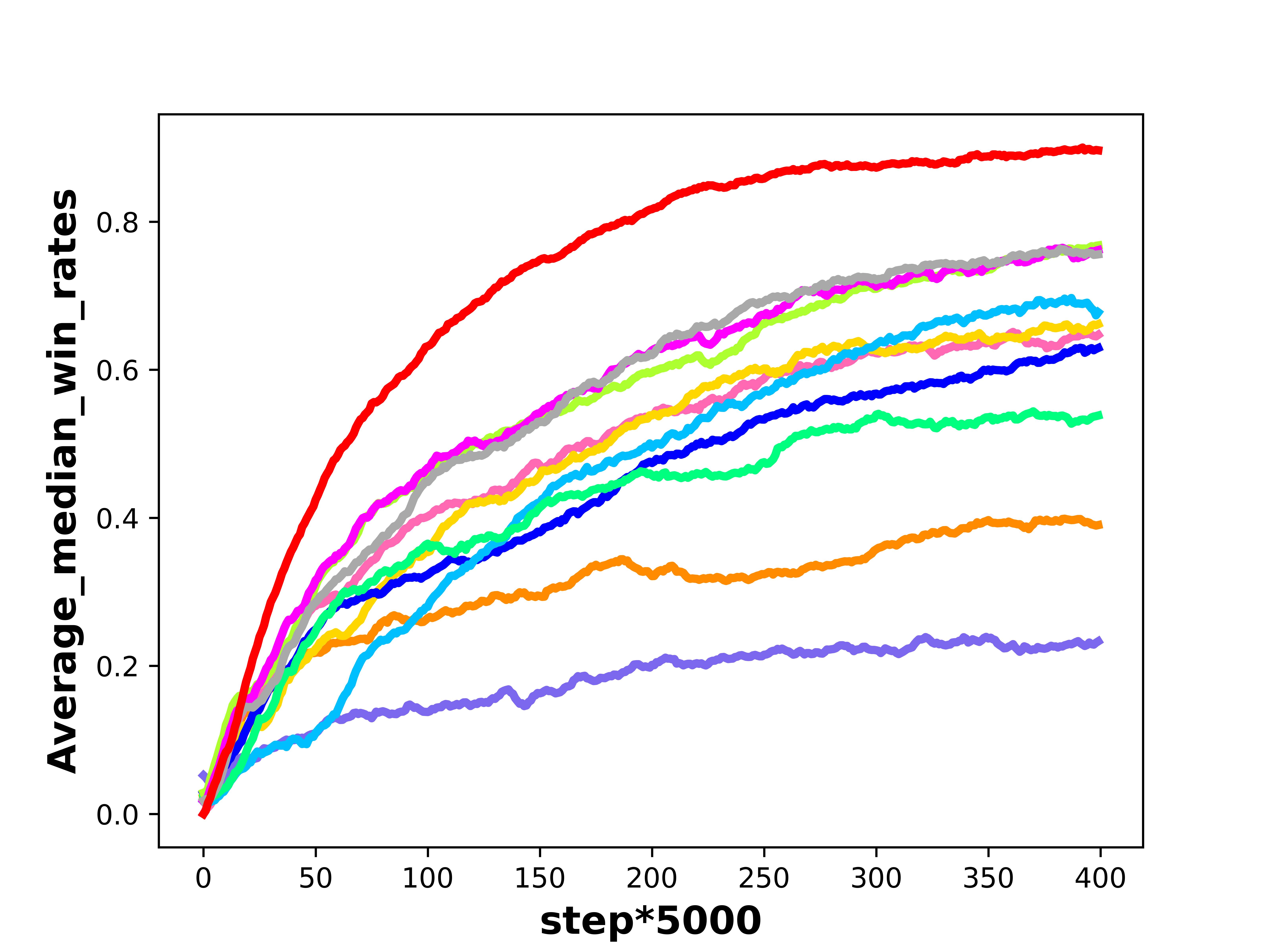}
    }
    \subfigure[Maps best out of 16 scenarios]{
        \includegraphics[width=6.5cm]{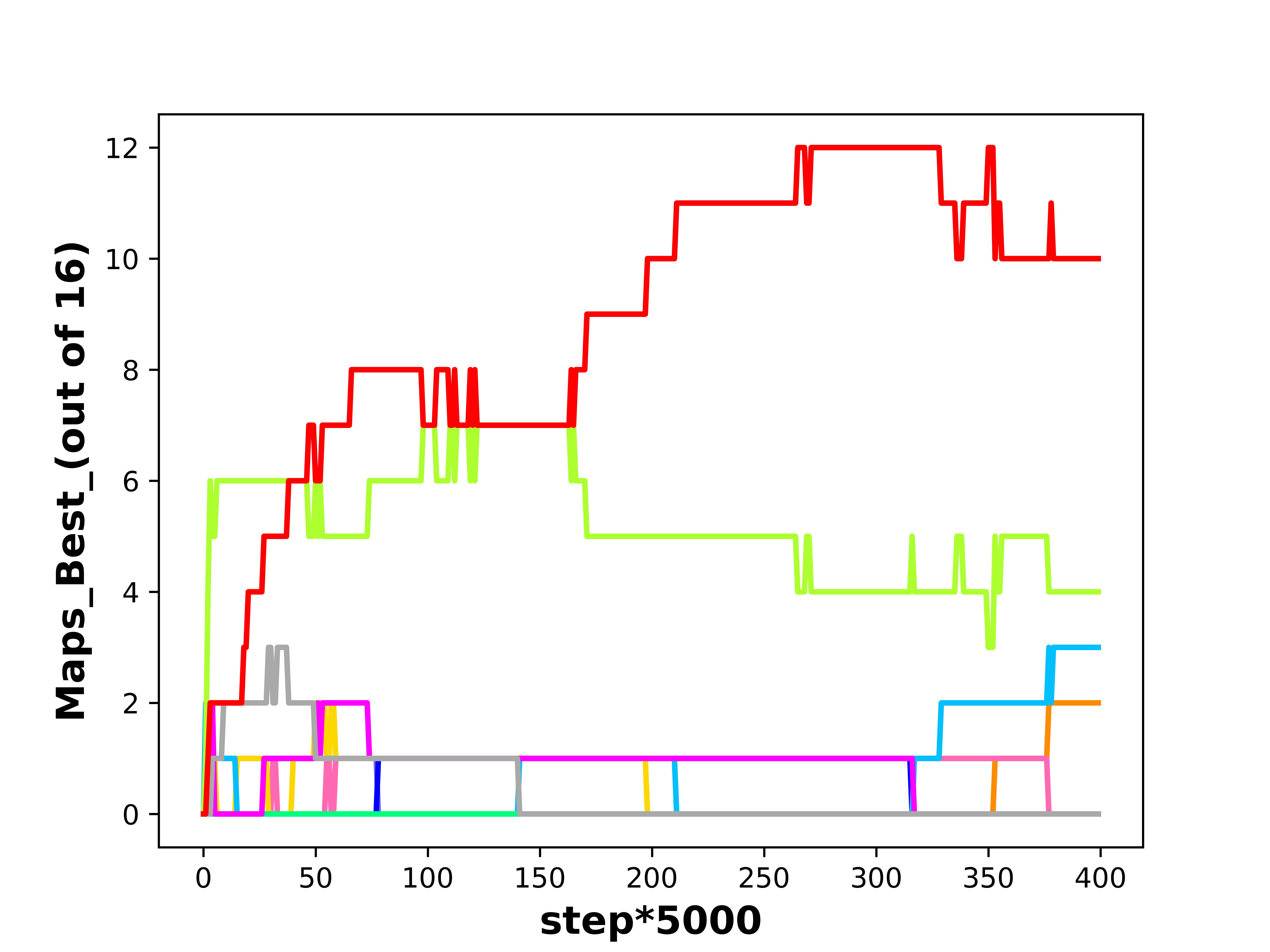}
    }
    \subfigure{
        \includegraphics[width=2.1cm]{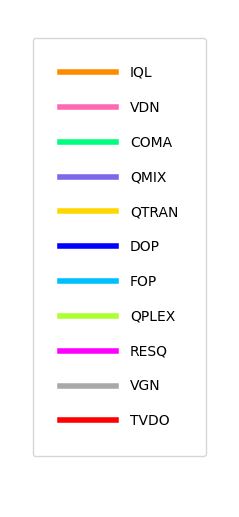}
    }
    \caption{(a) The median win rates, averaged across all 16 scenarios. Heuristic's performance is shown as a dotted line. (b) The number of scenarios, in which the median win rates of algorithms, is the highest by at least 1/32 (smoothed). }
    \label{fig:result_median_all_range}
\end{figure*}

The action space of each agent comprises four features: move direction, no-option, stop, and attack target. Deceased agents are restricted to selecting the no-option feature, whereas living agents are unable to choose it. Each agent has the option to either stop or move in any of the four cardinal directions: north, east, south, or west. Besides, the agent is permitted to execute the attack [enemy\_id] action only if the enemy is within the shooting range or field of attack.
The details of observation and action for each agent are shown in Table~\ref{table2}.

For all battle scenarios, the goal is to maximize the win rate and episode reward. Note that all agents obtain the same global reward, which is equivalent to the sum of the damage inflicted on all enemy agents collectively. The reward mechanism is that agents obtain 10 points for successfully eliminating an enemy unit. Additionally, a bonus of 200 points is awarded to each agent when they collectively eliminate all the enemies. The cumulative reward is normalized to within 20.

\begin{figure*}[t]
    \centering
    \subfigure[1c3s5z]{
        \includegraphics[width=4.1cm]{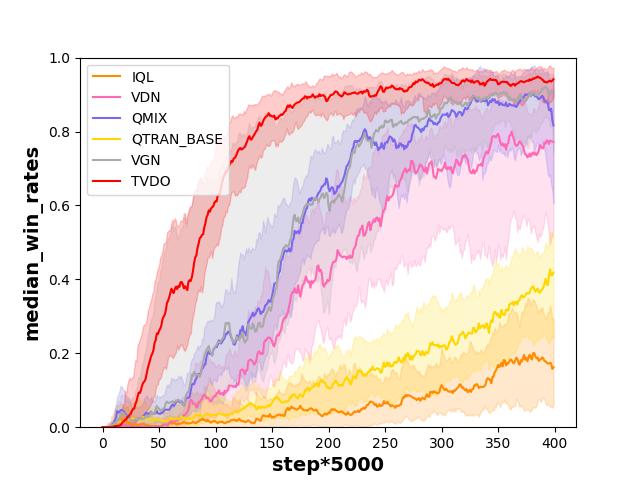}
    }
    \subfigure[2c\_vs\_64zg]{
        \includegraphics[width=4.1cm]{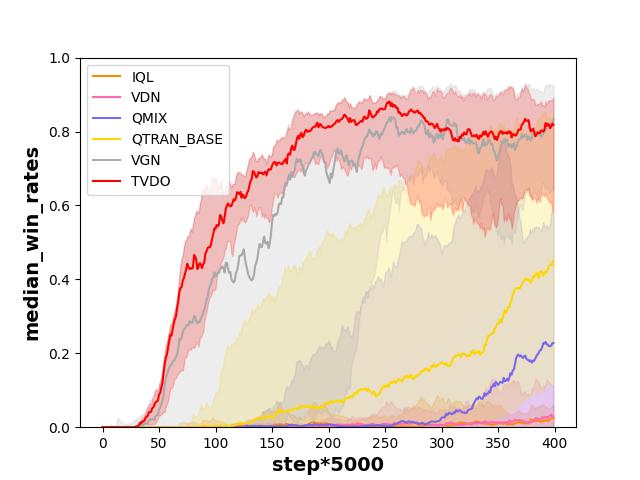}
    }
    \subfigure[2s\_vs\_1sc]{
        \includegraphics[width=4.1cm]{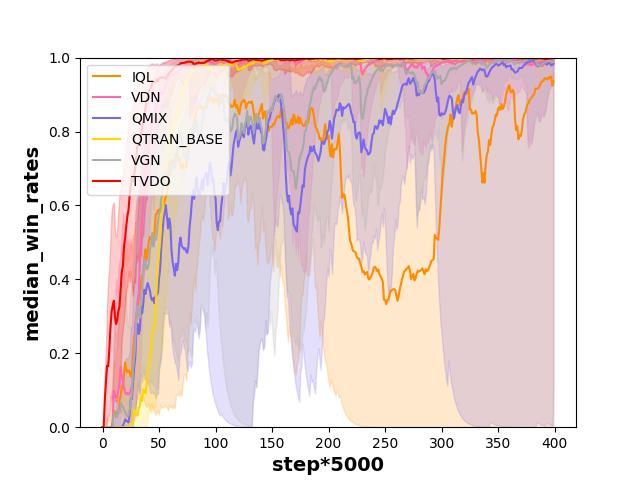}
    }
    \subfigure[2s3z]{
        \includegraphics[width=4.1cm]{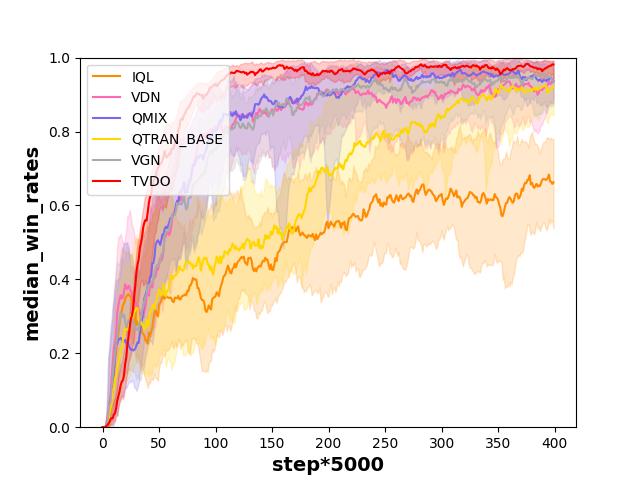}
    }
    \subfigure[3s\_vs\_3z]{
        \includegraphics[width=4.1cm]{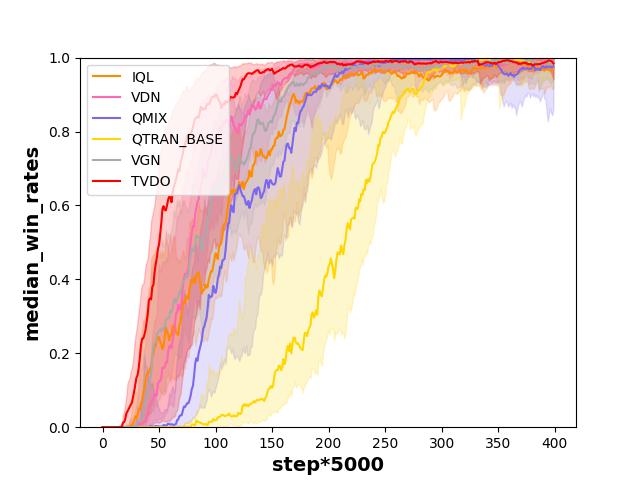}
    }
    \subfigure[3s\_vs\_5z]{
        \includegraphics[width=4.1cm]{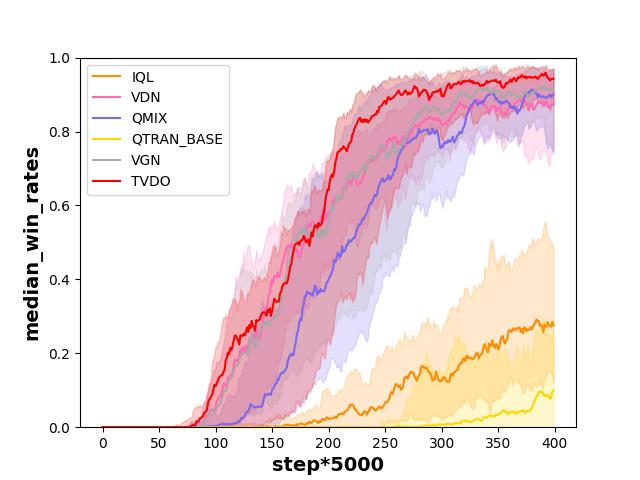}
    }
     \subfigure[3s5z]{
        \includegraphics[width=4.1cm]{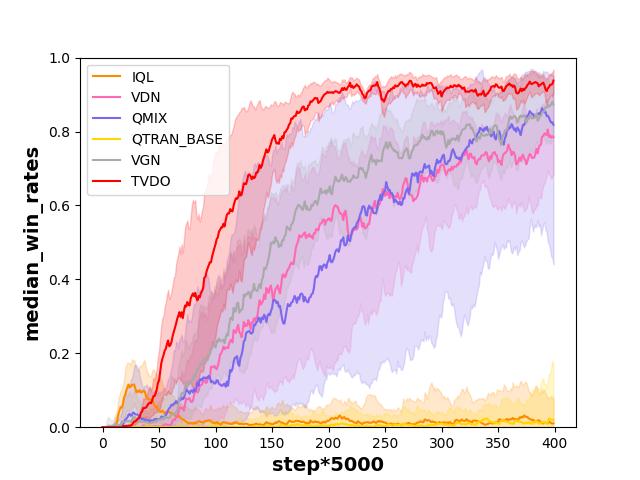}
    }
    \subfigure[5m\_vs\_6m]{
        \includegraphics[width=4.1cm]{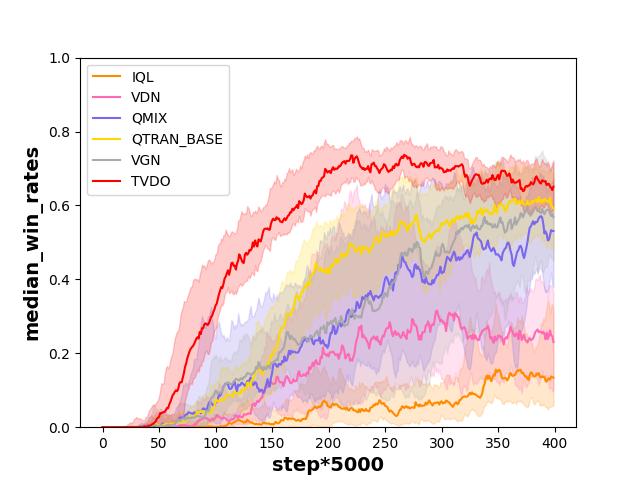}
    }
    \subfigure[6h\_vs\_8z]{
        \includegraphics[width=4.1cm]{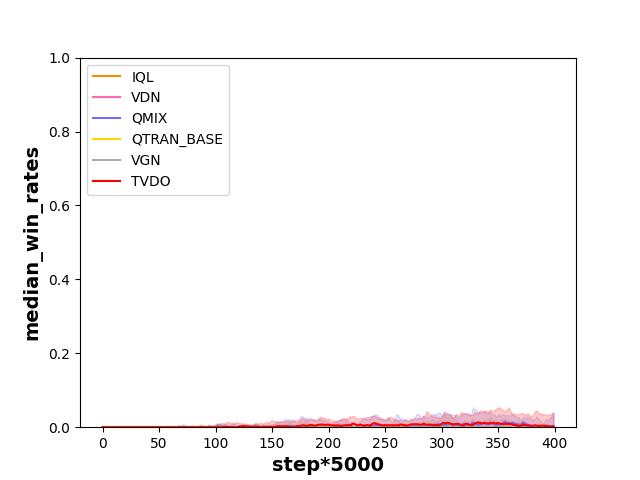}
    }
    \subfigure[8m]{
        \includegraphics[width=4.1cm]{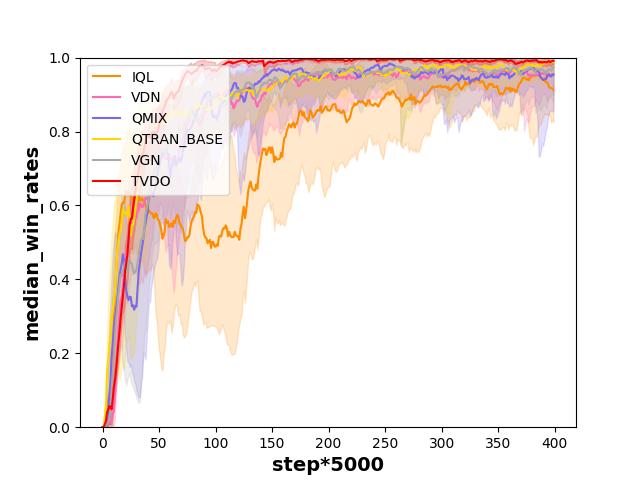}
    }
    \subfigure[10m\_vs\_11m]{
        \includegraphics[width=4.1cm]{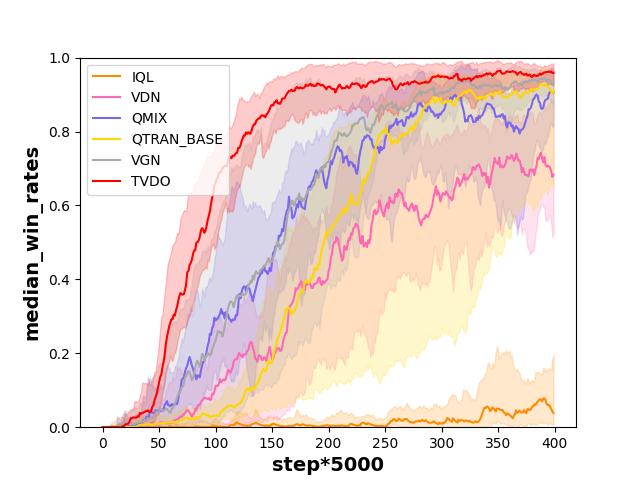}
    }
    \subfigure[25m]{
        \includegraphics[width=4.1cm]{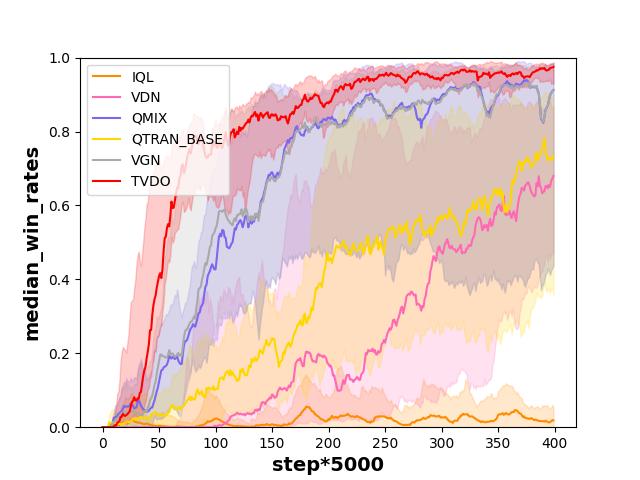}
    }
    \subfigure[27m\_vs\_30m]{
        \includegraphics[width=4.1cm]{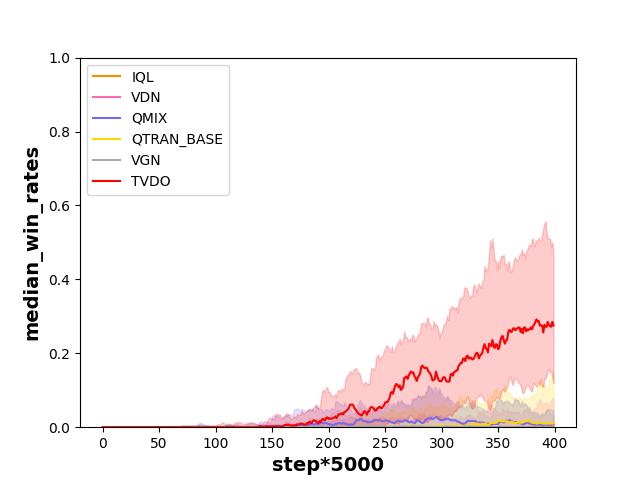}
    }
    \subfigure[bane\_vs\_bane]{
        \includegraphics[width=4.1cm]{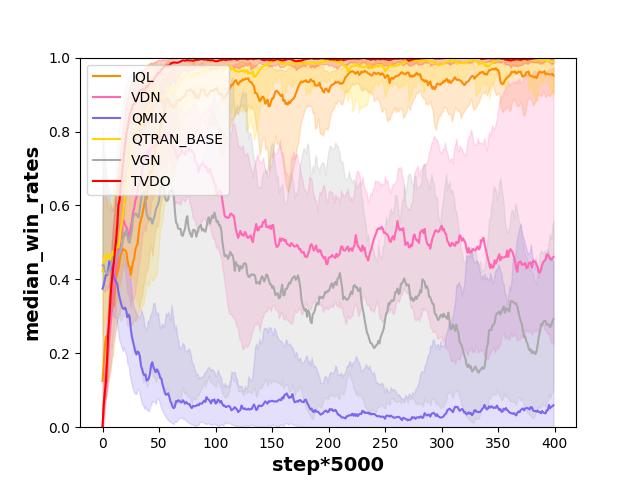}
    }
    \subfigure[MMM2]{
        \includegraphics[width=4.1cm]{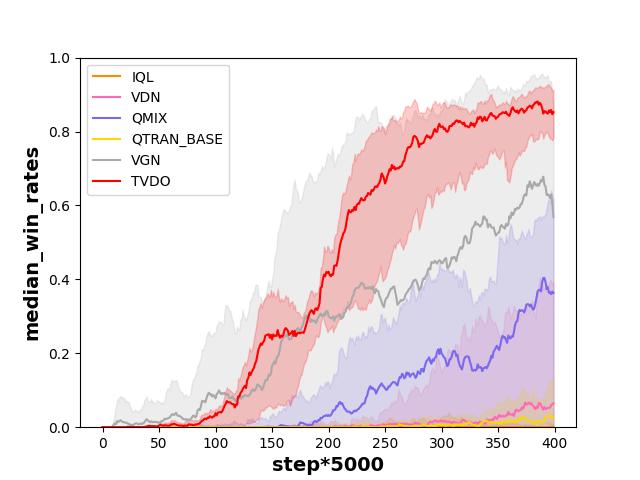}
    }
    \subfigure[so\_many\_baneling]{
        \includegraphics[width=4.1cm]{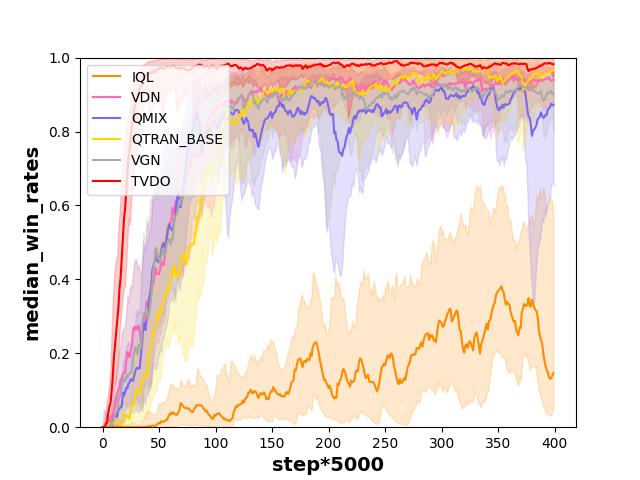}
    }
    \caption{The median win rates for TVDO and competing methods including IQL, VDN, QMIX, QTRAN, and VGN on environment of various difficulty in the SMAC benchmark with variance.}
    \label{fig:ivqq_range}
\end{figure*}

\begin{figure*}[t]
    \centering
    \subfigure[1c3s5z]{
        \includegraphics[width=4.1cm]{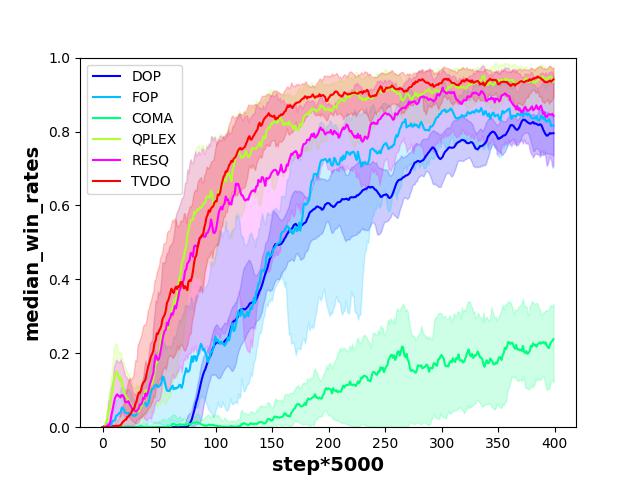}
    }
    \subfigure[2c\_vs\_64zg]{
        \includegraphics[width=4.1cm]{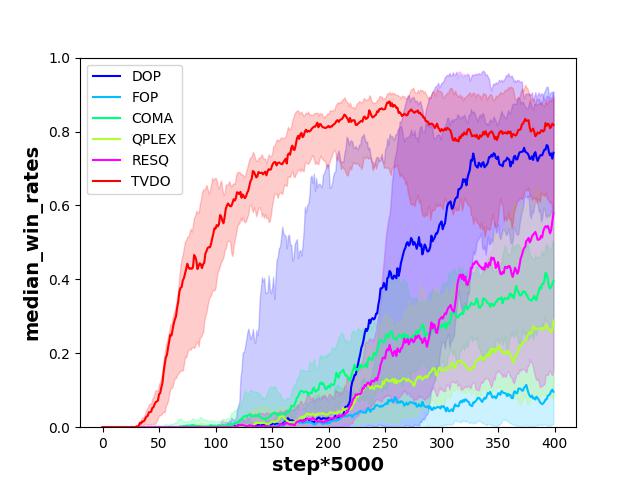}
    }
    \subfigure[2s\_vs\_1sc]{
        \includegraphics[width=4.1cm]{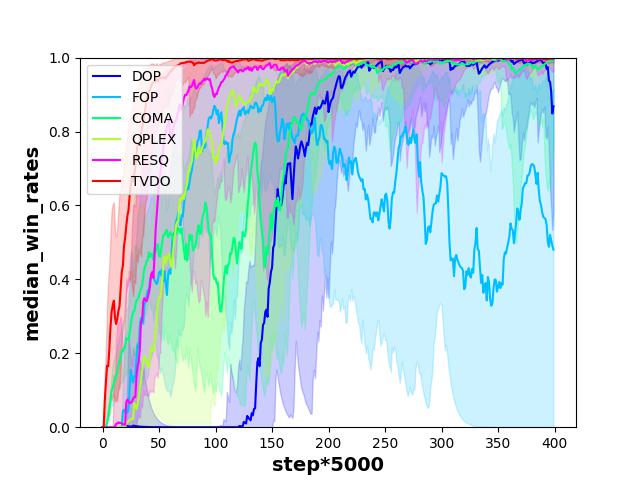}
    }
    \subfigure[2s3z]{
        \includegraphics[width=4.1cm]{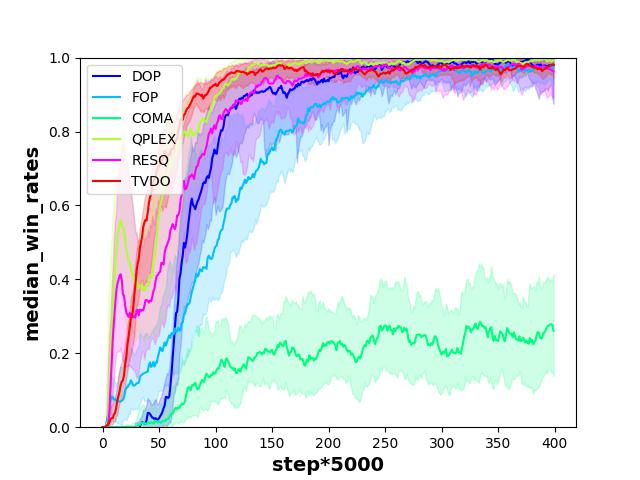}
    }
    \subfigure[3s\_vs\_3z]{
        \includegraphics[width=4.1cm]{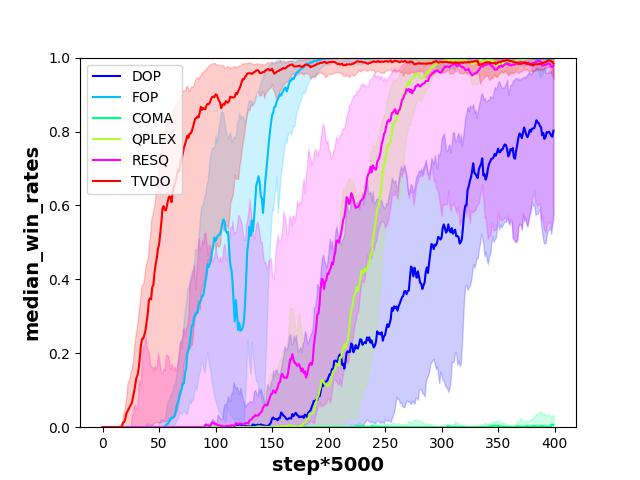}
    }
    \subfigure[3s\_vs\_5z]{
        \includegraphics[width=4.1cm]{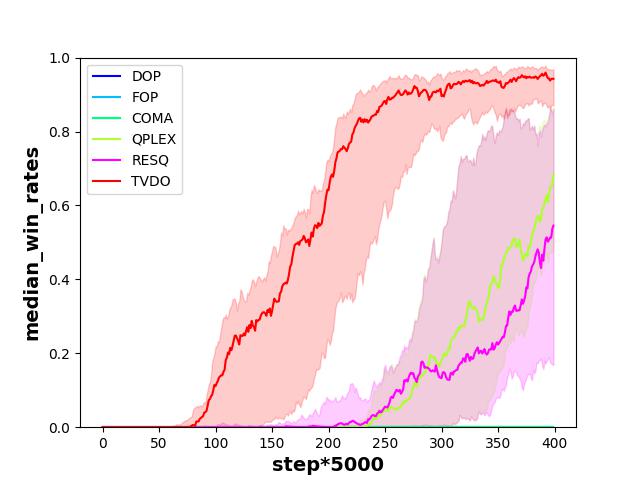}
    }
     \subfigure[3s5z]{
        \includegraphics[width=4.1cm]{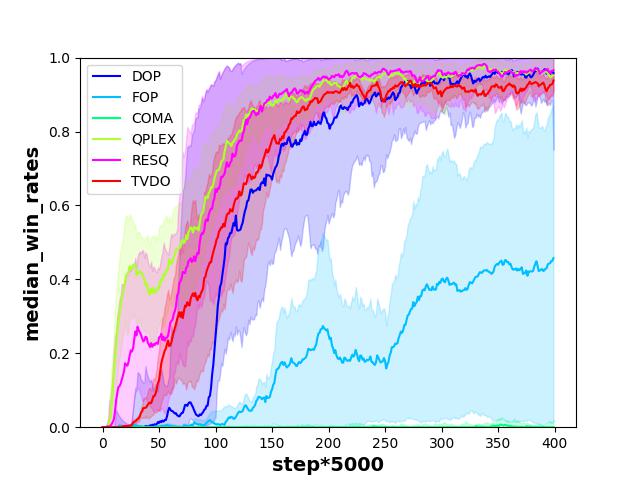}
    }
    \subfigure[5m\_vs\_6m]{
        \includegraphics[width=4.1cm]{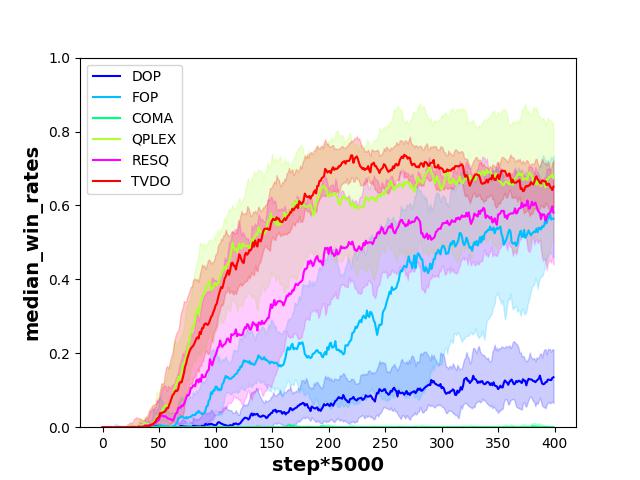}
    }
    \subfigure[6h\_vs\_8z]{
        \includegraphics[width=4.1cm]{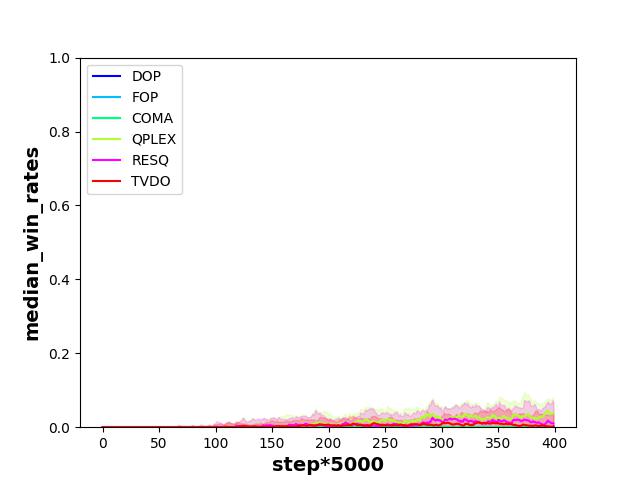}
    }
    \subfigure[8m]{
        \includegraphics[width=4.1cm]{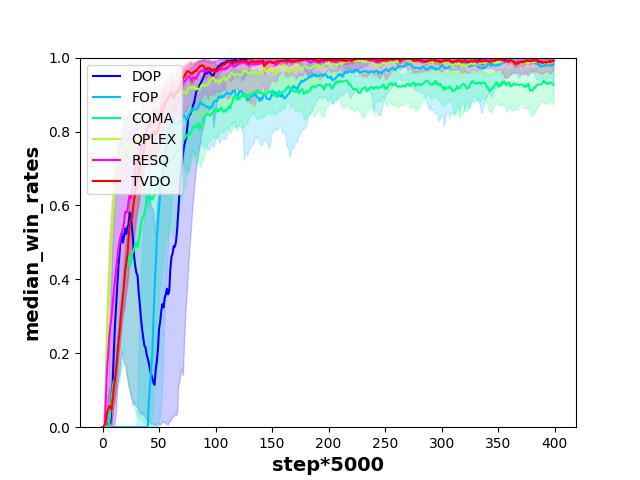}
    }
    \subfigure[10m\_vs\_11m]{
        \includegraphics[width=4.1cm]{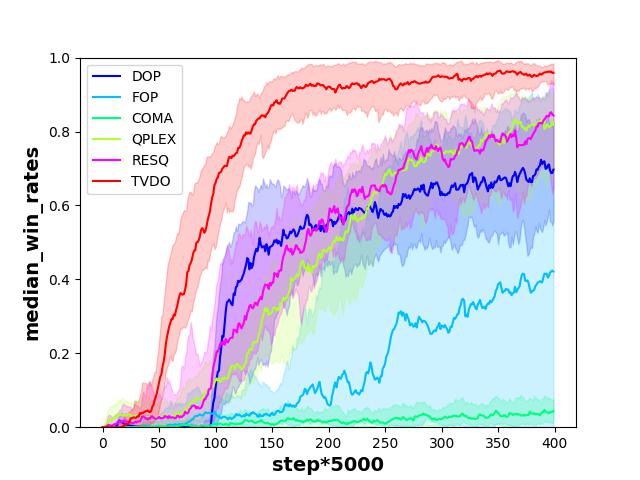}
    }
    \subfigure[25m]{
        \includegraphics[width=4.1cm]{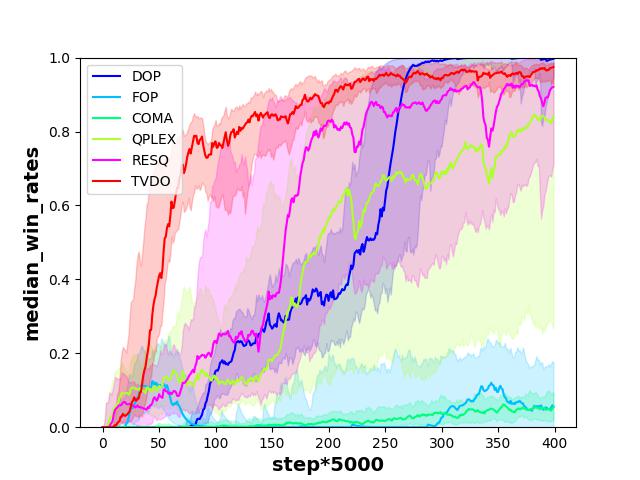}
    }
    \subfigure[27m\_vs\_30m]{
        \includegraphics[width=4.1cm]{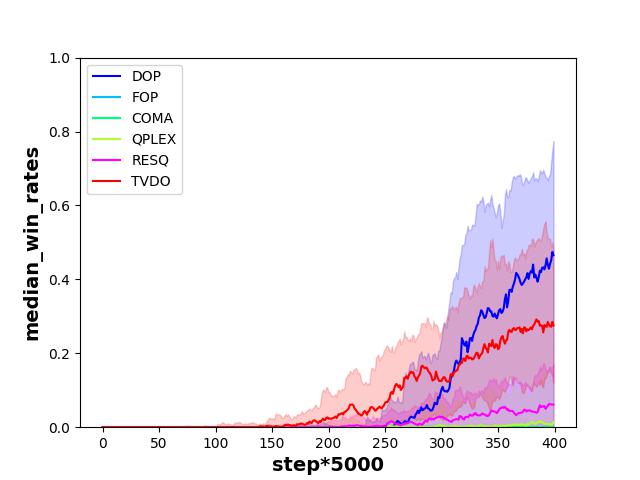}
    }
    \subfigure[bane\_vs\_bane]{
        \includegraphics[width=4.1cm]{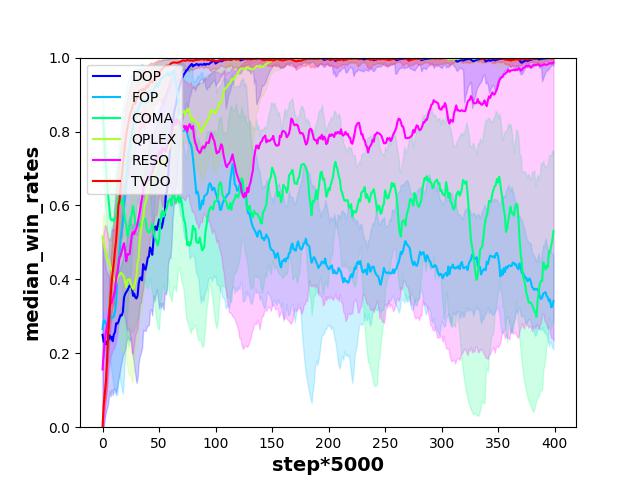}
    }
    \subfigure[MMM2]{
        \includegraphics[width=4.1cm]{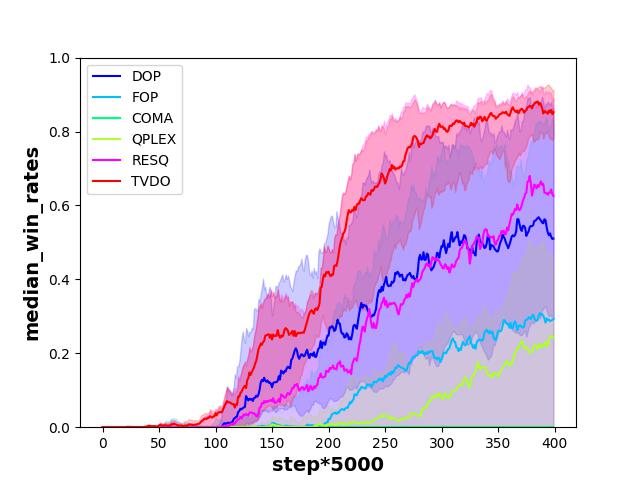}
    }
    \subfigure[so\_many\_baneling]{
        \includegraphics[width=4.1cm]{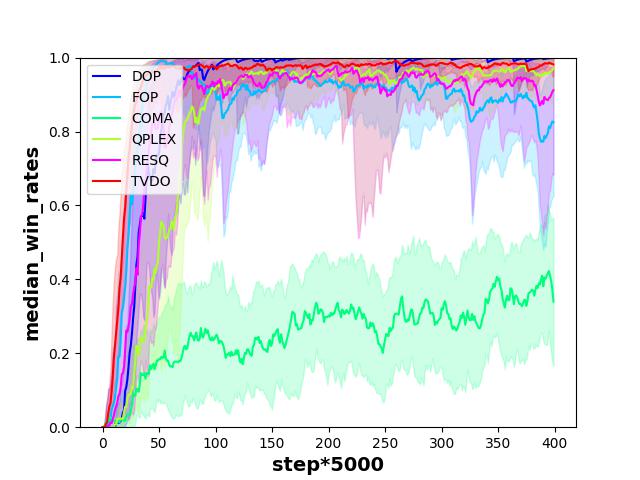}
    }
    \caption{The median win rates for TVDO and competing methods including COMA, DOP, FOP, QPLEX, and RESQ on environment of various difficulty in StarCraft \Rmnum{2} with variance.}
    \label{fig:dfcq_range}
\end{figure*}

\begin{table}[htbp]
    \vspace{-0.3cm}
    \centering
    \renewcommand{\arraystretch}{1.2}
    \caption{Hyperparameters for StarCraft \Rmnum{2} training.}
    \begin{tabular}{c c l}
        \hline
        \textbf{name} & \textbf{Value} & \textbf{Description}\\
        \hline
        optim & RMSprop & the optimizer of Torch \\
        difficulty & 7 & the difficulty of the game  \\
        n\_steps & $2 \times 10^6$ & Maximum steps until the end of training \\
        buffer\_size & 5000 & capacity of replay buffer \\
        batch\_size & 32 & number of samples from each update \\
        evaluate\_cycle & 5000 & how often to evaluate the model \\
        $lr$ & $5 \times 10^{-4}$ & learning rate \\
        $C$ & 200 & how often target networks update \\
        $\gamma$ & 0.99 & discount factor \\
        $\rho$ & 0.3 & the initial weight factor \\
        $\alpha$ & 0.999 & momentum coefficient \\
        \hline
    \end{tabular}
    \label{table3}
\end{table}

\subsection{Architecture and Training}
\label{sec:network}
The architecture of individual action-value function networks, which is shared across all baselines, comprises a Gated Recurrent Unit (GRU) combined with a fully connected layer before and after. The global action-value network $Q_\glb$ is composed of two hidden layers, each consisting of 64 units and ReLU non-linearities. Throughout the training, we use the ${\epsilon}$-greedy annealed with an initial value of 0.5 and a rate of 0.02 to select actions. Furthermore, $\gamma$ is set to 0.99. The replay buffer contains a collection of 5000 episodes. During the training process, batches with 32 episodes are sampled uniformly from the replay buffer. The training is conducted on fully unrolled episodes. After each episode, a single gradient descent step is performed to update the parameters of the networks. The setting of all hyperparameters for the training is presented in Table~\ref{table3}.

\subsection{Overall Results}
We compared the proposed TVDO method with ten SOTA MARL algorithms:
IQL~\cite{2017multiagent}, VDN~\cite{sunehag2017value}, COMA~\cite{foerster2018counterfactual}, QMIX~\cite{rashid2018qmix}, QTRAN~\cite{son2019qtran}, DOP~\cite{wang2020off}, QPLEX~\cite{wang2021qplex},  FOP~\cite{zhang2021fop}, RESQ~\cite{shen2022resq}, and VGN~\cite{Wei2024VGN}~\footnote{All comparison results are produced from the official codes: \\ 
FOP-\href{https://github.com/liyheng/FOP}{https://github.com/liyheng/FOP};
DOP-\href{https://github.com/TonghanWang/DOP}{https://github.com/TonghanWang/DOP}; 
QPLEX-\href{https://github.com/wjh720/QPLEX}{https://github.com/wjh720/QPLEX};
RESQ-\href{https://github.com/xmu-rl-3dv/ResQ}{https://github.com/xmu-rl-3dv/ResQ};
[IQL COMA VDN QMIX QTRAN]-\href{https://github.com/starry-sky6688/MARL-Algorithms}{https://github.com/starry-sky6688/MARL-Algorithms}.}. 
To evaluate the performance of each approach, we use the following evaluation procedure: after 5000 training steps, the training process is paused, and then an evaluation phase, where 32 episodes are executed in a greedy decentralized way, is initiated.
Furthermore, the win rates is defined as the proportion of these episodes where the agents eliminate all enemies within the limited steps.

To illustrate the overall performance of each method, Figure~\ref{fig:result_median_all_range} plots the averaged median win rates across all 16 scenarios, as well as the number of scenarios in which the algorithm outperforms. Meanwhile, we depict the performance of a basic heuristic method ignoring partial observability, in which each agent chooses the nearest enemy unit and engages in a coordinated attack with the entire team until the enemy agent is eliminated. Once the enemy agent is defeated, then the agent selects the next closest enemy unit to attack. Employing the fundamental form of focus-firing is a significant strategy to achieve favorable performance in various scenarios. However, the comparatively inferior performance of the heuristic method indicates that SMAC tasks require more sophisticated strategies beyond simple focus-firing. 

As shown in Figure~\ref{fig:result_median_all_range} (a), we can observe that the proposed TVDO method significantly and constantly outperforms all baselines and exhibits median win rates that are higher than 15\% on average across 16 scenarios. Moreover, VDN, QMIX, QTRAN, QPLEX, and TVDO almost outperform COMA and FOP, illustrating the sample efficiency of value-based MARL approaches compared to policy-based MARL methods expect DOP. This is because DOP combines off-policy tree backup updates with the on-policy TD($\lambda$) technique to solve the issue of lower training efficiency. Additionally, Figure~\ref{fig:result_median_all_range} (b) illustrates that emerges as the top performer across a maximum of twelve scenarios. We can find that the number of maps that TVDO performs best gradually decreases to 10 after 1.5M steps. This is because all baselines achieve almost 100\% win rates in some \textit{Easy} scenarios, such as 2s\_vs\_1sc and 8m. 

Furthermore, QTRAN performs well in the climb and penalty game but poorly in most scenarios of SMAC benchmark compared with our method and some MARL baselines as shown in Figure~\ref{fig:result_median_all_range}. It suggests that there may be some challenges when using QTRAN to solve some more complicated tasks due to the extra limitation, i.e., affine transformation. In contrast, our proposed TVDO method without any extra constraints achieves significant improvement in the performance of convergence speed and stability.

\subsection{Comparison Results}
All comparison results are shown in Figure \ref{fig:ivqq_range} and \ref{fig:dfcq_range}. 
From these experimental results, we have several aspects of observations. Firstly, TVDO is noticeably the strongest performer in all of the scenarios, in particular on the maps with heterogeneous agents. The largest performance gap can be seen on the 10m\_vs\_11m, 2c\_vs\_64zg, and MMM2. Because these asymmetric scenarios require learning a policy that has precise control to consistently defeat the enemy, which indicates that the superior representational capacity of TVDO presents a clear benefit over other value-decomposition methods. 

Secondly, TVDO and most MARL approaches including VDN, QMIX, DOP, QPLEX, RESQ, and VGN achieve reasonable performance on easy maps, which shows the advantage of learning the factorized action-value functions. However, almost all baselines perform not well on \textit{Hard} and \textit{Super-Hard} scenarios. In specific, the result on the 2c\_vs\_64zg scenario, which contains 2 Colossi allied units and 64 Zerglings enemy units, presents that only TVDO and DOP can easily find the winning strategy, as shown in Figure~\ref{fig:ivqq_range} (b) and Figure~\ref{fig:dfcq_range} (b). Furthermore, in the \textit{Super-Hard} task MMM2, TVDO achieves the best performance (nearly 90\% win rates) among all methods. These scenarios have a common feature that the quantity of enemy units is larger than the number of allied units, which requires the method to combine these combat units to form powerful tactics and strategies, and then achieve victory in the battle. 

Last, in the scenario of 5m\_vs\_6m and 3s\_vs\_5z, TVDO reaches good performance, while other baselines perform quite differently. In particular, 5m\_vs\_6m, consisting of 5 allied marines and 6 enemy marines, is an asymmetric task that requires precise control such as keeping a distance and evading attacks from the enemy units to win consistently. As observed in Figure~\ref{fig:ivqq_range} (h) and Figure~\ref{fig:dfcq_range} (h), TVDO and QPLEX significantly outperform other baselines with higher sample efficiency. However, QPLEX performs poorly on some maps such as 3s\_vs\_5z, 2c\_vs\_64zg, and MMM2, which should suffer from the extra constraint during the process of learning factorized action-value function decomposition. 

\subsection{Learned Policies}
To better understand the differences between the learned policy by each method, we examine the learned behavior of each agent from the battle replay~\footnote{Demonstrative videos are available at \href{https://sites.google.com/view/tvdo}{https://sites.google.com/view/tvdo}.}.
On the symmetric scenario with stalker and zealot units, i.e., 2s3z and 3s5z, these approaches with not good performance, such as VDN, COMA, and FOP, learn a specific strategy that agents initially move left and then engage enemies once they are in the shooting range, without considering other factors such as enemy position or unit weaknesses. However, TVDO learns a positioning strategy that allied stalkers are protected from enemy zealots, which is achieved by coordinating the behaviors of allied stalkers and zealots. Concretely, the allied zealots are instructed to block off the enemy zealots, preventing them from directly attacking the stalkers, and then the allied stalkers can fire at the enemy from a safe distance. It indicates that our method takes into account more factors such as units' capabilities and uses them strategically to maximize the chance of success.

On the homogeneous scenarios with marines units, such as 8m, 25m, and 10m\_vs\_11m, VDN, QMIX, and QTRAN learn a basic coordinated policy known as focus-firing by having the allied agents focus on a single enemy unit to eliminate it quickly. Although this simple strategy performs well on the 8m, it becomes difficult to achieve victory for the scenario with numerous agents, i.e., 25m, or asymmetric map such as 10\_vs\_11m. In contrast, TVDO can consistently learn an intricate strategy that positions allied agents into a semicircle to attack enemy units from the sides, which leads to higher cumulative rewards. Meanwhile, allied marines adopt their unique skill that using stimpack injectors to self-administer stimulants to increase the attack speed and movement speed for this strategy. It demonstrates that TVDO can learn sophisticated tactics and policy including using special skills of agents to defeat the enemy.

On the bane\_vs\_bane scenario with a substantial amount of enemy and allied units, i.e., 20 zerglings and 4 banelings, VDN and FOP learn initially an elementary policy of directly firing the visible enemy units, and then engage in more exploratory behaviors, such as attempting to move instead of attacking, which results in a significant decline in performance due to agents may make a suboptimal decision. The banelings possess superior combat capabilities but require strategic protection, which makes them used to break through defensive lines or deal devastating blows to an opponent's army. TVDO and some baselines including QPLEX and DOP learn how to combine movement and firing together and build lines of defense against the enemy, which indicates that the method needs to leverage agents' advantages and adopt flexible tactics to continuously improve control ability, then win in battle.

\subsection{Limitation}
To show the limitations of the proposed method, we conduct experiments on the 6h\_vs\_8z scenario, which is categorized as \textit{Super-Hard}. As shown in Figure~\ref{fig:ivqq_range}-\ref{fig:dfcq_range} (i), TVDO as well as all baselines fail to solve the task. Particularly, in the 6h\_vs\_8z scenario including 6 Hydralisks and 8 Zealots, one winning strategy is requiring all Hyralisks to ambush enemy units in their path and then attack together when the Zealots approach. Since all approaches employ noise-based exploration, the agents of the team face challenges in identifying states that are worth exploring and struggle to effectively coordinate their exploration efforts towards those states. Without improved exploration techniques, we found it to be extremely challenging for any methods to pick up this strategy~\cite{liu2021cooperative}. It demonstrates that efficient exploration for MARL is still a challenging and open problem.

\section{Conclusion and Future Work}
In this paper, we proposed a novel MARL method called factorized Tchebycheff value-decomposition optimization (TVDO) to keep the consistency of jointly-trained policies and individually-executed actions. Particularly, a nonlinear Tchebycheff aggregation function was formulated to realize the global optimum by tightly constraining the upper bound of individual action-value bias, which is inspired by the Tchebycheff method of multi-objective optimization. Then, our theoretical analysis demonstrated that TVDO could precisely express the value decomposition with a guarantee of consistency between global and individual policies. Empirically, in the climb and penalty game, we verified that TVDO could represent precisely the global-to-individual value factorization with a guarantee of policy consistency, and it also achieved a significant performance superiority over some SOTA MARL baselines in the SMAC benchmark. 

In the future, we will combine some cooperative exploration techniques with our method to improve the performance in more complex and real scenarios. Furthermore, one promising future direction is to integrate some techniques, including distributed optimization or offline simulation training.


\appendices
\section{The Derivation of the Bounds for Weight Factor}
\label{app:derivation}
\subsection{The derivation of the upper bound for Weight Factor}
In the proof of sufficiency of Theorem~\ref{the:TVDO}, we note that $Q_\glb(\tau, \overline{u}) \geq Q_{\glb}(\tau, u)$ need to be proved. According to the Eq.~(\ref{con:TVDO1}), we can show that
\begin{equation}
    Q_\glb(\tau,\overline{u}) = \sum_{i=1}^{N}{Q_i(\tau_i,u_i)} + \sum_{i=1}^{N}{\{ \left| Q_i(\tau_i,\overline{u}_i) - Q_i(\tau_i,u_i)\right| \}}.
\end{equation}
According to the Eq.~(\ref{con:TVDO1}), we can show that
\begin{equation}
    \begin{split}
        Q_{\glb}&(\tau, u) \leq \sum_{i=1}^{N}{Q_i(\tau_i,u_i)} + \rho E(\tau, u) \\
        &= \sum_{i=1}^{N}{Q_i(\tau_i,u_i)} + \rho \max\limits_{1\leq i \leq N}{\{ \left| Q_i(\tau_i,u_i) - Q_i(\tau_i,\overline{u}_i) \right| \}}.
    \end{split}
\end{equation}
In order to ensure that the $Q_\glb(\tau, \overline{u}) \geq Q_{\glb}(\tau, u)$ holds, we only need to guarantee that $\sum_{i=1}^N{\{ \left| Q_i(\tau_i,u_i) - Q_i(\tau_i,\overline{u}_i) \right| \}}$ is greater than or equal to $\rho \cdot \max\limits_{1\leq i \leq N}{\{ \left| Q_i(\tau_i,u_i) - Q_i(\tau_i,\overline{u}_i) \right| \}}$. In other word, $\rho$ only need to be less than or equal to $\frac{\sum_{i=1}^N{\{ \left| Q_i(\tau_i,u_i) - Q_i(\tau_i,\overline{u}_i) \right| \}}}{\max\limits_{1\leq i \leq N}{\{ \left| Q_i(\tau_i,u_i) - Q_i(\tau_i,\overline{u}_i) \right| \}}}$. 

\subsection{The derivation of the lower bound for Weight Factor} 
In the proof of necessity of Theorem~\ref{the:TVDO}, we note that $\Gamma \geq 0$ need to be proved. According to the definition of IGM condition, the equation $Q_{\glb}(\tau, \overline{u}) = \max\limits_{u}{Q_\glb(\tau, u)} \geq Q_\glb(\tau,u)$ holds. Then we can show that 
\begin{equation}
    \Gamma \geq \sum_{i=1}^{N}{Q_i(\tau_i,u_i)} - Q_\glb(\tau, \overline{u}) + \rho E(\tau,u).
\end{equation}
To guarantee the equation $\Gamma \geq 0$ holds, we only need to satisfy $\rho E(\tau, u) \geq Q_{\glb}(\tau, \overline{u}) - \sum_{i=1}^{N}{Q_i(\tau_i, u_i)}$. Then we can show that 
\begin{equation}
    \begin{split}
        \rho E(\tau, u) &= \rho \max\limits_{1\leq i \leq N}{ \{ \left| Q_i(\tau_i,u_i) - Q_i(\tau_i,\overline{u}_i) \right| \} } \\
        &\geq Q_{\glb}(\tau, \overline{u}) - \sum_{i=1}^{N}{Q_i(\tau_i, u_i)}.
    \end{split}
\end{equation}
Therefore, the weight factor $\rho$ should be greater than or equal to $\frac{Q_{\glb}(\tau, \overline{u}) - \sum_{i=1}^{N}{Q_i(\tau_i, u_i)}}{\max\limits_{1\leq i \leq N}{\{ \left| Q_i(\tau_i,u_i) - Q_i(\tau_i,\overline{u}_i) \right| \}}}$.

\section{The Proof of the Bounds for Weight Factor}
\label{app:proof}
According to the Theorem~\ref{the:TVDO}, the weight factor $\rho$ should be limited in a certain range for conforming the IGM condition, formally,
\vspace{-0.3cm}
\begin{equation}
    \left\{
        \begin{array}{c c}
             \vspace{0.2cm}
             \frac{Q_{\glb}(\tau, \overline{u}) - \sum_{i=1}^{N}{Q_i(\tau_i, u_i)}}{\max\limits_{1\leq i \leq N}{\{ \left| Q_i(\tau_i,u_i) - Q_i(\tau_i,\overline{u}_i) \right| \}}} \leq \rho. & \\
             \frac{\sum_{i=1}^N{\{ \left| Q_i(\tau_i,u_i) - Q_i(\tau_i,\overline{u}_i) \right| \}}}{\max\limits_{1\leq i \leq N}{\{ \left| Q_i(\tau_i,u_i) - Q_i(\tau_i,\overline{u}_i) \right| \}}} \geq \rho.
        \end{array}
    \right.
\end{equation}

Let $a$ and $b$ denote the lower and upper bound of $\rho$, respectively. When $u_i\rightarrow\overline{u}_i$, there should be $\sum_{i=1}^{N}{Q_i(\tau_i,u_i)} \geq Q_\glb(\tau,u)$ for the accumulated value factorization. Therefore, we can show that
\begin{equation}
    \begin{split}
        b &= \frac{\sum_{i=1}^N{\{ \left| Q_i(\tau_i,u_i) - Q_i(\tau_i,\overline{u}_i) \right| \}}}{\max\limits_{1\leq i \leq N}{\{ \left| Q_i(\tau_i,u_i) - Q_i(\tau_i,\overline{u}_i) \right| \}}} \\
        &= \frac{\sum_{i=1}^N{ Q_i(\tau_i,\overline{u}_i) } - \sum_{i=1}^N{Q_i(\tau_i,u_i)}}{\max\limits_{1\leq i \leq N}{\{ \left| Q_i(\tau_i,u_i) - Q_i(\tau_i,\overline{u}_i) \right| \}}} \\
        &\geq \frac{ Q_{\glb}(\tau, \overline{u}) - \sum_{i=1}^N{Q_i(\tau_i,u_i)}}{\max\limits_{1\leq i \leq N}{\{ \left| Q_i(\tau_i,u_i) - Q_i(\tau_i,\overline{u}_i) \right| \}}} = a.
    \end{split}
\end{equation}
Therefore, the lower bound is always lower than or equal to the upper bound.

\bibliographystyle{unsrt}  
\bibliography{bare_jrnl_new_sample4}  

\begin{IEEEbiography}[{\includegraphics[width=1in,height=1.25in, clip,keepaspectratio]{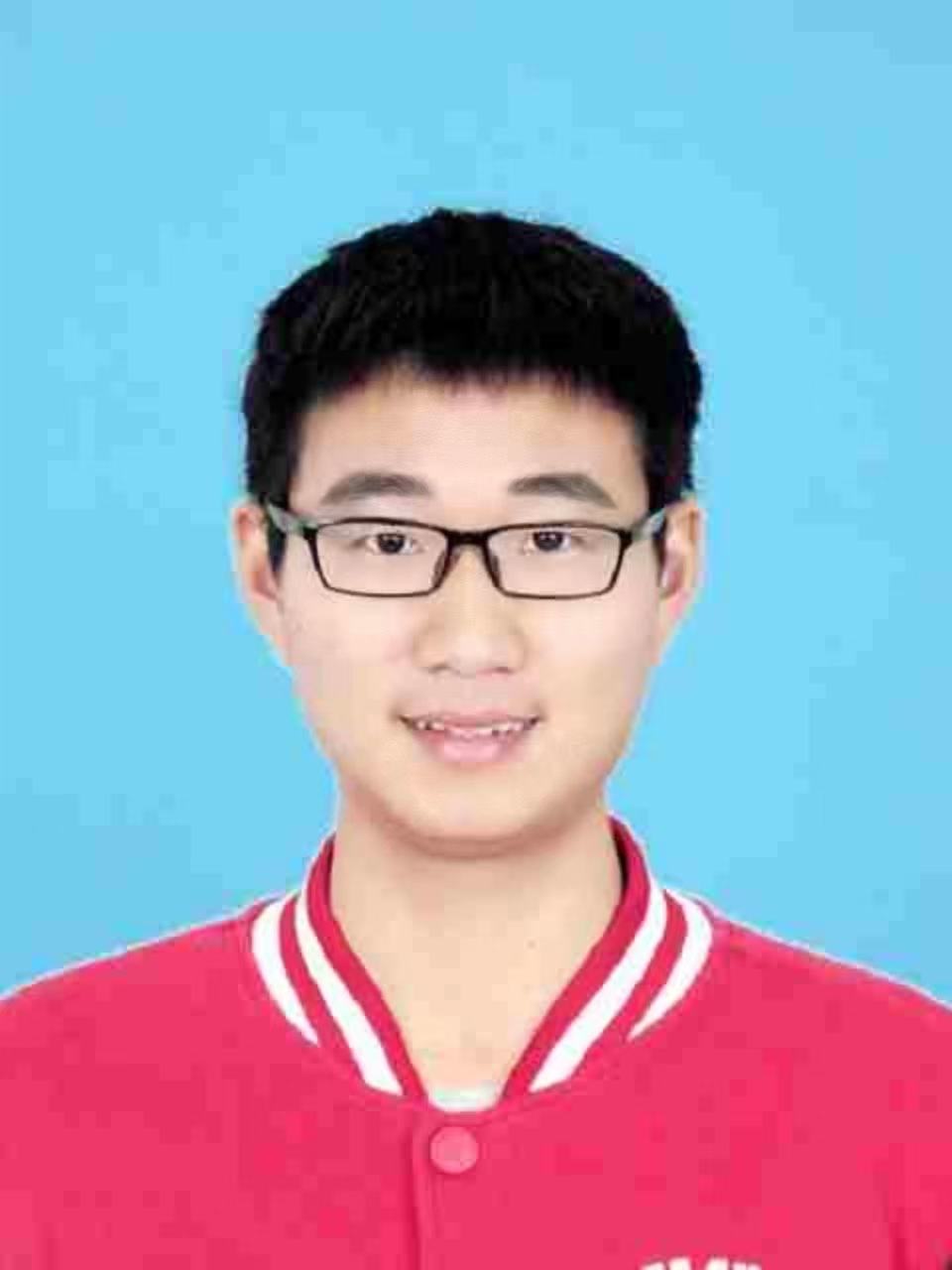}}]{Xiaoliang Hu} received the B.S. degree in Jimei University in 2018, and M.S degree in School of Computer Science and Engineering at Huaqiao University in 2021. He is currently working toward a Ph.D. degree at the PCALab, Key Lab of Intelligent Perception and Systems for High-Dimensional Information of the Ministry of Education, School of Computer Science and Engineering, Nanjing University of Science and Technology. 

His research interests include multi-agent reinforcement learning, multi-agent systems, optimal control, and their industrial applications.
\end{IEEEbiography}

\begin{IEEEbiography}[{\includegraphics[width=1in,height=1.25in,clip,keepaspectratio]{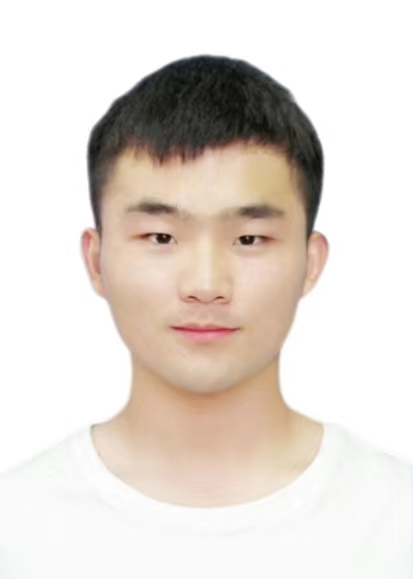}}]{Pengcheng Guo} received  B.S. degree in DaLian Ocean University in 2020, and M.S. degree at the PCALab, Key Lab of Intelligent Perception and Systems for High-Dimensional Information of Ministry of Education, School of Computer Science and Engineering, Nanjing University of Science and Technology in 2024. 

His research interests include multi-agent reinforcement learning and machine learning.
\end{IEEEbiography}

\begin{IEEEbiography}[{\includegraphics[width=1in,height=1.25in,clip,keepaspectratio]{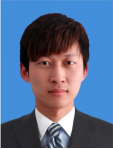}}]{Yadong Li} received the M.S. degree from China University of Mining and Technology in 2018 and currently working toward Ph.D. degree at PCALab, Key Lab of Intelligent Perception and System for High-Dimensional Information of Ministry of Education, School of Computer Science and Engineering, Nanjing University of Science and Technology. 

Since 2019, he has been a lecturer with the School of Information Science and Engineering, Zaozhuang University. 
His research interests include pattern recognition, machine learning, and deep reinforcement learning.
\end{IEEEbiography}

\begin{IEEEbiography}[{\includegraphics[width=1in,height=1.25in,clip,keepaspectratio]{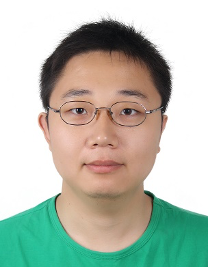}}]{Guangyu Li} (Member IEEE) received the B.S. degree from China University of Mining and Technology and M.S. degree from Tongji University, China, in 2008 and 2011, respectively, and the Ph.D. degree from University of Paris-Sud, Paris, France, in 2015. He is currently working as an associate professor with the Key Laboratory of Intelligent Perception and Systems for High-Dimensional Information of Ministry of Education, Nanjing University of Science and Technology, Nanjing, China. 
His current research interests include machine learning, reinforcement learning, computer vision, wireless networks, etc.
\end{IEEEbiography}

\begin{IEEEbiography}[{\includegraphics[width=1in,height=1.25in,clip,keepaspectratio]{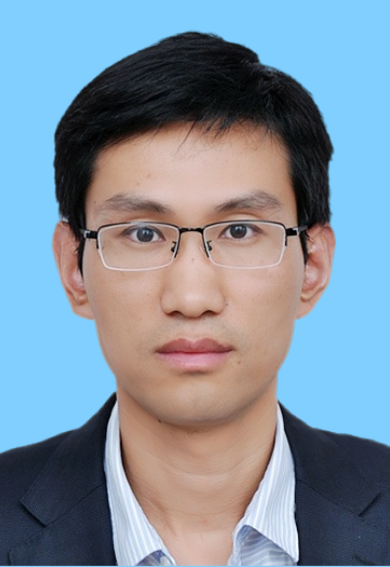}}]{Zhen Cui} (Member IEEE) received the B.S. degree from Shandong Normal University, Jinan, China, in 2004, the M.S. degree from Sun Yat-sen University, Guangzhou, China, in 2006, and the Ph.D. degree from Institute of Computing Technology (ICT), Chinese Academy of Sciences, Beijing, China, in 2014. 
He was a Research Fellow in the Department of Electrical and Computer Engineering at National University of Singapore (NUS) from Sep 2014 to Nov 2015. He also spent half a year as a Research Assistant on Nanyang Technological University (NTU) from Jun 2012 to Dec 2012. Currently, he is a Professor of Nanjing University of Science and Technology, China.
His research interests cover pattern recognition and machine learning, especially focusing on graph deep learning, deep reinforcement learning, multi-agent reinforcement learning, etc.
\end{IEEEbiography}

\begin{IEEEbiography}[{\includegraphics[width=1in,height=1.25in,clip,keepaspectratio]{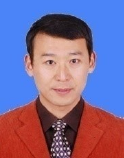}}]{Jian Yang} received the PhD degree from Nanjing University of Science and Technology (NJUST), on the subject of pattern recognition and intelligence systems in 2002. 

In 2003, he was a Postdoctoral researcher at the University of Zaragoza. From 2004 to 2006, he was a Postdoctoral Fellow at Biometrics Centre of Hong Kong Polytechnic University. From 2006 to 2007, he was a Postdoctoral Fellow at Department of Computer Science of New Jersey Institute of Technology. Now, he is a Chang-Jiang professor in the School of Computer Science and Technology of NUST. He is the author of more than 200 scientific papers in pattern recognition and computer vision. His papers have been cited more than 5000 times in the Web of Science, and 13000 times in the Scholar Google. His research interests include pattern recognition, computer vision, and machine learning. 

Currently, he is/was an associate editor of \textit{Pattern Recognition}, \textit{Pattern Recognition Letters}, \textit{IEEE Trans. Neural Networks and Learning Systems}, and \textit{Neurocomputing}. He is a Fellow of IAPR.
\end{IEEEbiography}

\vfill

\end{document}